\documentclass[a4paper,11pt]{article}
\usepackage{jheppub} 
\makeatletter
\renewcommand{\@fpheader}{\phantom{Prepared for submission to JHEP}}
\makeatother
\usepackage{tikz}
\usepackage{enumitem}
\usepackage{graphicx}
\usepackage{tabularx} \usetikzlibrary{arrows.meta,decorations.pathmorphing,decorations.pathreplacing,calc,shapes.geometric}
\definecolor{accent}{RGB}{30,90,160}
\definecolor{bingray}{RGB}{236,236,236}
\definecolor{hlyellow}{RGB}{255,246,205}
\usepackage{amsmath,amssymb,amsthm}
\usepackage{mathrsfs}
\theoremstyle{plain}
\newtheorem{theorem}{Theorem}
\newtheorem{proposition}{Proposition}

\newtheorem{lemma}{Lemma}
\theoremstyle{definition}
\newtheorem{definition}[theorem]{Definition}
\theoremstyle{plain}

\theoremstyle{remark}

\title{The Page Transition in the Multi-Time Correlations of Hawking Radiation}

\author{Tushar Waghmare}
\affiliation{%
Department of Physics and Astronomy,\\
University of Manchester,\\
Oxford Road, Manchester M13 9PL, United Kingdom
}
\emailAdd{tushar.waghmare@postgrad.manchester.ac.uk}

\abstract{The Page curve and the island formula describe a reorganization of quantum information in Hawking radiation across the Page time. We study how this reorganization appears in correlations among emissions collected at several retarded times. The late radiation is treated as a covariant multi-time process, or process tensor, and
\(\chi_{\mathrm{island}}=\mathcal N_\varnothing-\mathcal N_R\)
measures the temporal memory removed when the early radiation register \(R\) is supplied. Operationally, this quantity is the corresponding change in the optimal distinguishability of the radiation process from a Markovian reference. We exhibit evaporation processes whose single-bin states and all process marginals below a finite order \(k_\ast\geq3\) remain fixed and coincide with a CP-divisible reference, while \(\chi_{\mathrm{island}}\) changes by order one at order \(k_\ast\). The Page transition can therefore appear as a transition in the Markov order of the radiation process. For mediated evaporation, entanglement-wedge reconstruction of the predictive interior from \(R\) Markovianizes the late process; conversely, when the late process has a minimal predictive memory and its future-affecting information is mediated by that memory, Markovianization by \(R\) supplies a simulator of its influence. A Hayden--Preskill evaporation comb and an evaporating-JT calculation realize this mechanism. The late radiation retains temporal memory across the transition, while the early radiation becomes sufficient for predicting its future multi-time statistics.
}

\begin{document}
\maketitle
\flushbottom

\section{Introduction}
\label{sec:introduction}

The Page curve describes the transfer of quantum information from an
evaporating black hole to its Hawking radiation \cite{PageInformation}.  In
semiclassical gravity, the quantum-extremal-surface transition and the island
formula reproduce this curve and place part of the black-hole interior in the
entanglement wedge of the collected radiation
\cite{EngelhardtWall,Penington,AEMM,AMMZ,PSSY}.  Entanglement-wedge
reconstruction then states that interior operators can be represented on the
radiation within a code subspace \cite{Penington,CotlerUniversal}.  These
results identify where the information is encoded and when the interior
becomes reconstructible.

The Page curve and the island formula have sharpened one central aspect of this problem by identifying how the entropy of the Hawking radiation evolves and when the entanglement wedge of the radiation acquires an interior component \cite{Penington,AEMM}. This establishes where information can be encoded and, through entanglement-wedge reconstruction, when interior information becomes recoverable from the radiation. Unitary evaporation, however, also constrains the dynamical organization of that information. An asymptotic observer receives Hawking quanta successively over retarded time, so the radiation defines both a sequence of states and a hierarchy of correlations among emissions at different times. A complete operational account of the Page transition should therefore relate the geometric reorganization of the entanglement wedge to a corresponding reorganization of these time-resolved correlations: how information carried by the black hole influences later radiation, how that influence is distributed across successive emissions, and what portion of the previously collected radiation is sufficient to predict its future effects.

This dynamical question is complementary to the entropy and reconstruction viewpoints. Generalized-entropy calculations determine the dominant saddle associated with a quantum extremal surface for a radiation region, while reconstruction identifies the bulk information accessible from that region. The same entropy curve can nevertheless be compatible with distinct temporal correlation structures, and information that is invisible in individual emissions can reside in correlations extending across several emission times.

The relevant question is therefore when the early radiation becomes sufficient side information for predicting the subsequent Hawking process. In other words, when conditioning on the early radiation screens the interior's influence on correlations among later emissions. Characterizing this change requires an operational object sensitive to the joint, multi-time response of the radiation rather than only to its state at a single time. The process tensor, or quantum comb, provides such a description: it assigns joint outcome probabilities to sequences of temporally ordered interventions and allows correlations to be classified by the temporal span required to detect them \cite{Pollock,CDP}. We use this multi-time description to formulate the Page transition as a reorganization of the predictive memory carried by the Hawking radiation.

We divide the radiation available to the observer into an early register
\(R\) and late bins \(L=\{b_1,\ldots,b_n\}\).  The late process can be tested
after discarding \(R\), or with \(R\) retained as side information.  Let
\(\mathcal N_\varnothing\) and \(\mathcal N_R\) denote the corresponding
operational distances from the factorized Markov-comb class.  The change in memory produced by access to the early radiation is
\begin{equation}
  \chi_{\mathrm{island}}(t)
  =\mathcal N_\varnothing(t)-\mathcal N_R(t).
  \label{eq:intro-chi}
\end{equation}
It quantifies the reduction in late-radiation temporal memory resulting from
conditioning on \(R\).  Both terms are
radiation observables: they are obtained from multi-time tests on the
collected radiation, with the second family of tests allowed to process the
early register.  The operational distance also gives
\(\chi_{\mathrm{island}}\) a decision-theoretic reading as the reduction in
optimal distinguishability from the Markovian reference class produced by
access to \(R\).

The mechanism is predictive reconstruction.  A minimal interior memory
\(M_{\min}\) carries the information from the past that affects future late
emissions.  In the evaporation class considered here, the early--late
correlations are mediated by this interior history, and conditioning on the
predictive memory renders the late process approximately Markovian.  Once
entanglement-wedge reconstruction supplies a channel from \(R\) to
\(M_{\min}\), the early radiation can reproduce the influence of the
interior on subsequent emissions.  The Page transition is then read as the
acquisition of predictive sufficiency by \(R\).

Our first result localizes the transition at a finite temporal order.  We
construct a family of evaporation processes for which every single-bin state
is thermal and unchanged, and every reduced process involving fewer than
\(k_\ast\) time steps agrees with a fixed Markovian reference.  At order
\(k_\ast\geq3\), the conditioned-memory observable changes from a small value
to an order-one value.  The minimal example has \(k_\ast=3\) and carries a
shared parity correlation among three late bins.  The Page transition can
therefore be invisible to one-time thermality and single-step divisibility
while remaining visible to a finite multi-time experiment.

The second result relates this transition to interior reconstruction.  For
mediated evaporation, reconstruction of the predictive interior from \(R\)
on a code subspace bounds the residual \(R\)-conditioned memory by the
mediation error, the memory remaining after conditioning on the interior, and
the reconstruction error.  The converse runs in the predictive direction:
when the late process has a minimal sufficient memory and no independent
future-affecting side channel, Markovianization by \(R\) yields a simulator of
that memory's influence on all admissible future tests.  Together, these two
directions establish \(\chi_{\mathrm{island}}\) as a witness of island-like
reconstruction within the stated class, while preserving the distinction
between recovery of predictive information and reconstruction of an
arbitrarily enlarged microscopic interior register.

The final result expresses the transition as a memory budget.  The conditional
memory cost is the smallest auxiliary register that, together with \(R\),
renders the late process Markovian to a chosen tolerance.  Before recovery it
is set by the minimal predictive memory; after recovery the trivial register
suffices.  This cost supplies a resource interpretation of the transition and
motivates a semiclassical comparison with the generalized-entropy budget of
the quantum extremal surface.

We realize these statements in a finite-dimensional evaporation family
inspired by Hayden--Preskill decoding \cite{HaydenPreskill}.  A two-bin
screening core gives the conditioned-memory transition in closed form, and a
three-bin parity extension fixes its first visible order.  Numerical tests
vary the detector family, code-subspace size, temporal resolution, and memory
dimension.  These variations change the resolved scale and the recovery
accuracy while preserving the same transition in predictive sufficiency.

The semiclassical realization uses an evaporating JT black hole coupled to a
bath, with an infalling diary serving as the predictive interior memory.  The
order-three witness becomes a four-entropy combination evaluated with the
island prescription.  The regions containing \(R\) exchange between exterior
and diary-containing saddles at different retarded times, producing a finite
process-transition interval rather than a single completion time
\cite{AEMM,Hollowood:2020cou,Hollowood:2020couPage}.  Once the relevant
islands are compatible, the geometric terms cancel in the conditional
combination and the remaining quantity is a bulk conditional mutual
information, corrected by the finite QES gaps and the reconstruction error.
This gives the gravitational form of the screening mechanism.

In summary, the Page curve and the island prescription determine the entropy and entanglement wedge of the radiation. Entanglement-wedge reconstruction and recoverability determine how interior information can be represented on the radiation. Process tensors determine how that representation changes the predictions of an ordered radiation experiment.  Our previous work supplied the covariant process construction and the operational memory measure \cite{priorpaper}; the present paper identifies the Page-time reorganization, its finite Markov order, and its predictive interpretation.

The paper is organized as follows.  Section~\ref{sec:framework} introduces
the covariant radiation process and its operational memory.  Section~\ref{sec:conditioned}
defines the early-radiation conditioning and the memory observables.
Section~\ref{sec:results} states the finite-order, reconstruction, converse,
and memory-cost results.  Sections~\ref{sec:model} and~\ref{sec:numerics}
give the evaporation comb and its robustness checks.
Section~\ref{sec:semiclassical} supplies the JT and QES realization, and
Section~\ref{sec:discussion} develops the process wedge, the semiclassical
memory budget, and the algebraic formulation.  The appendices collect the
proofs, numerical certification, semiclassical bookkeeping, and notation.

\section{The covariant process-tensor framework}
\label{sec:framework}

An asymptotic observer receives Hawking radiation as an ordered sequence of
coarse-grained emissions.  The state of one radiation bin determines the
statistics of a measurement performed at one retarded time.  Temporal memory is
a property of the joint response across several times: an operation performed
on one accessible portion of the radiation can alter the conditional statistics
assigned to later portions, either through correlations already present in the
field or through an ancillary record retained by the observer.  The process
tensor is the operational object that contains these multi-time probabilities.

To describe these multi-time probabilities, let $b_1,\ldots,b_n$ be the
radiation bins encountered along the collection worldline. At slot $i$, the
observer chooses an instrument
$\mathsf{A}^{(i)}=\{\mathsf{A}^{(i)}_{a_i}\}_{a_i}$ from an admissible family
$\mathcal{I}$.  Each element is a completely positive map, and the sum over its
outcomes is trace preserving.  The maps include the preparation and readout of
a localized detector and may pass an ancillary quantum or classical record to
later slots; they do not require radiation to be sent back toward the black
hole.  The late-radiation process tensor $\Upsilon_L$ is the multilinear
functional that returns the joint outcome probabilities,
\begin{equation}
  p(a_1,\ldots,a_n\mid
    \mathsf{A}^{(1)},\ldots,\mathsf{A}^{(n)})
  =
  \Upsilon_L\!\left[
    \mathsf{A}^{(1)}_{a_1},\ldots,
    \mathsf{A}^{(n)}_{a_n}
  \right].
  \label{eq:processtensor}
\end{equation}
Equivalently, $\Upsilon_L$ is a quantum comb whose open legs are contracted
with the Choi operators of the instrument elements through the link product
\cite{Pollock,CDP}.  Its causality constraints ensure that summing over all
outcomes at a slot removes any dependence on instrument choices made later.
In a finite-dimensional truncation, an informationally complete family
reconstructs the comb.  For the radiation field, $\Upsilon_L$ denotes the
operational equivalence class of combs restricted to the detector algebra and
the instruments in $\mathcal{I}$.  The accessible memory is therefore stated
relative to this family.

To obtain this comb from the radiation field, we use the covariant detector construction of
Ref.~\cite{priorpaper}.  The observer's collection worldline is covered by ordered,
partly overlapping causal diamonds, with a coarse retarded-time interval and an
accessible radiation algebra assigned to each diamond.  A localized probe
couples to the field inside the corresponding diamond; its preparation,
interaction, and readout realize one element of the instrument at that slot.
The field evolution and the retained detector ancillas link successive slots,
and their composition gives $\Upsilon_L$ on the ordered bins
$b_1,\ldots,b_n$.  Once the worldline, detector supports, coarse graining, and
instrument family have been fixed, no additional choice of spacetime foliation
enters the process.  This is the sense in which the multi-time radiation
observable is covariant.

To identify temporal memory, we compare the resulting process with a Markovian
reference. A process is Markovian when the present output is a sufficient state for the
future at every slot.  In the comb representation, its Choi operator factorizes
as
\begin{equation}
  \Upsilon^{\mathrm{M}}_{n:0}
  =
  \rho_0\otimes
  J[\Lambda_{1:0}]\otimes\cdots\otimes
  J[\Lambda_{n:n-1}],
  \label{eq:markovfactorization}
\end{equation}
up to the ordering of the input and output legs, with each
$\Lambda_{i:i-1}$ completely positive and trace preserving.  The map advancing
one step is then independent of the earlier intervention history once the
state delivered at the preceding slot is specified.  We denote this full
factorized reference class by $\mathrm{CPD}$, retaining the notation of the
prior work.  Throughout this manuscript, $\mathrm{CPD}$ refers to the
Markovian combs satisfying~\eqref{eq:markovfactorization}; it is stronger than
CP-divisibility of a family of reduced two-time dynamical maps.  The latter is a
lower-order diagnostic and can hold even when the full process tensor contains
multi-time memory.

The reference class is generally nonconvex.  A convex mixture of factorized
combs can retain a common classical label across the slots and thereby generate
temporal correlations.  The optimization below is therefore performed over
the factorized combs themselves, rather than over their convex hull, so that
zero memory is equivalent to genuine process-tensor Markovianity.

We measure departure from this reference class through the experiments
available to the observer. A tester $T\in\mathcal{I}$ is an adaptive sequence of admissible instruments,
with later choices allowed to depend on earlier outcomes and with a final
measurement on the retained record.  For a base divergence $D$ on the tester
outcomes, define the operational comb distance
\begin{equation}
  d_{\mathrm{op}}(\Upsilon,\Lambda)
  =
  \sup_{T\in\mathcal{I}}
  D\!\left(T[\Upsilon] \,\middle\|\, T[\Lambda]\right).
  \label{eq:dop}
\end{equation}
The trace-distance choice gives the optimal bias in symmetric comb
discrimination, while the relative-entropy choice gives the asymptotic
type-II error exponent under the corresponding tester class
\cite{CDP,GutoskiWatrous}.  The distance obeys data processing under every
admissible comb-to-comb channel $\mathcal{E}$,
\begin{equation}
  d_{\mathrm{op}}\!\left(
    \mathcal{E}[\Upsilon],\mathcal{E}[\Lambda]
  \right)
  \le
  d_{\mathrm{op}}(\Upsilon,\Lambda).
  \label{eq:dopdpi}
\end{equation}
The temporal memory accessible to the observer is its distance from the
Markovian reference set,
\begin{equation}
  \mathcal{N}[\Upsilon]
  =
  \inf_{\Lambda\in\mathrm{CPD}}
  d_{\mathrm{op}}(\Upsilon,\Lambda).
  \label{eq:Ndef}
\end{equation}
In the finite-dimensional families used below the infimum is attained.  This
definition compares processes through the experiments they support, rather
than through a divergence between untested Choi operators.  Computable state
or outcome-level quantities can therefore serve as witnesses or exact
specializations of~\eqref{eq:Ndef} once an instrument family and a base
divergence are fixed.

\begin{figure}[t]
\centering
\begin{tikzpicture}[
  scale=1.05, font=\small,
  bdy/.style={line width=1pt},                                  
  hor/.style={line width=0.9pt, dashed},                        
  rec/.style={-{Latex[length=2.6mm]}, line width=1pt, dashed, color=accent},
]

\draw[bdy] (0,0) -- (0,12);          
\draw[bdy] (0,12) -- (6,6);          
\draw[bdy] (6,6) -- (0,0);           
\node[above] at (0,12) {$i^{+}$};
\node[right]  at (6,6)  {$i^{0}$};
\node[below]  at (0,0)  {$i^{-}$};
\node[rotate=-45] at (1.85,10.78) {$\mathscr{I}^{+}$};
\node[rotate=45]  at (4.35,3.55)  {$\mathscr{I}^{-}$};

\fill[black!16] (3.5,3.35) -- (0,6.85) -- (0,7.15) -- (3.5,3.65) -- cycle;
\draw[line width=0.5pt, black!45] (3.5,3.5) -- (0,7);

\fill[black!8] (0,7) -- (0,9.5) -- (2.5,9.5) -- cycle;            
\draw[decorate, decoration={zigzag, segment length=5pt, amplitude=2.2pt},
      line width=1pt] (0,9.5) -- (2.5,9.5);                       
\node[above] at (1.0,9.62) {singularity};
\draw[hor] (0,7) -- (2.5,9.5);                                    
\node[rotate=45] at (1.55,7.95) {horizon};

\fill[accent!22] (0,7.4) -- (0.95,8.6) -- (0,9.2) -- cycle;
\node at (0.32,8.4) {$M$};
\fill[accent] (0.95,8.6) circle (1.4pt);
\node[anchor=west] at (1.08,8.74) {QES};

\draw[decorate, decoration={brace, amplitude=5pt, mirror, raise=3pt},
      line width=0.7pt] (5.4,6.6) -- (4.4,7.6);
\node at (5.18,7.42) {$R$};

\fill (4.0,8.0) circle (1.3pt);
\draw[line width=0.7pt] (3.85,7.85) -- (4.15,8.15);              
\node[anchor=west] at (4.30,8.18) {$t_{\mathrm{Page}}$};

\foreach \cx/\cy/\lab in {3.6/8.4/{b_1}, 3.2/8.8/{b_2}, 2.8/9.2/{b_3}}{%
  \draw[fill=black!12, draw=black!55, line width=0.5pt]
    (\cx,\cy+0.16) -- (\cx+0.16,\cy) -- (\cx,\cy-0.16) -- (\cx-0.16,\cy) -- cycle;
}
\node[anchor=west] at (3.84,8.60) {$b_1$};
\node[anchor=west] at (3.44,9.00) {$b_2$};
\node[anchor=west] at (3.04,9.40) {$b_3$};

\draw[rec] (4.6,7.4) to[out=215, in=315] (0.82,8.56);
\node[align=center, text=accent] at (2.62,6.10)
  {$\mathcal{R}_{R\to M}$\\[-1pt]\footnotesize EWR,\ \ $\epsilon(C)$};

\end{tikzpicture}
\caption{Conformal diagram of an evaporating black hole. Hawking radiation reaches
$\mathscr{I}^{+}$, where the early register $R$ (earlier retarded time, nearer $i^{0}$)
and the late bins $b_1,b_2,b_3$ are separated by the Page point $t_{\mathrm{Page}}$.
The interior memory $M$ is the island bounded by the quantum extremal surface. After
$t_{\mathrm{Page}}$ the early radiation reconstructs $M$ through $\mathcal{R}_{R\to M}$,
entanglement-wedge reconstruction with error $\epsilon(C)$ on a code subspace $C$.}
\label{fig:schematic}
\end{figure}
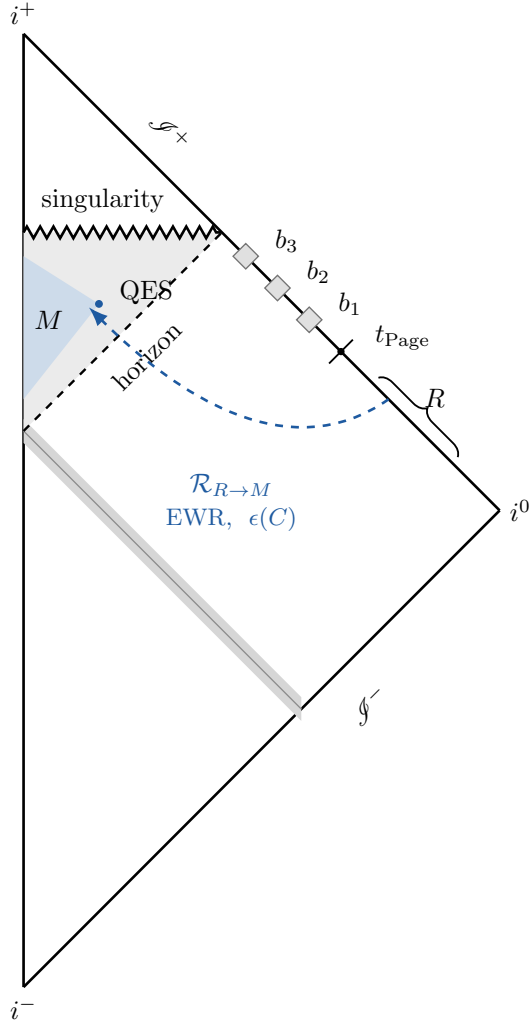

To determine when memory becomes visible, we can also restrict the tester
class by the number of slots it spans. A process
may coincide with a fixed Markovian comb on every reduced process involving at
most $K$ slots and still differ at a higher temporal order.  Single-bin
thermality, adjacent-pair process tensors, and CP-divisibility of the reduced
one-step maps then remain unchanged, while a longer adaptive experiment
resolves the hidden correlation.  Quantum Markov order formalizes this
instrument-dependent screening structure \cite{Taranto}: it identifies the
smallest temporal history that, under the stated instruments, renders the
future conditionally independent of the more distant past.  The resulting
hierarchy distinguishes the full multi-time memory of the radiation from every
fixed low-order projection of it.

\section{Island-conditioned processes and the order parameter}
\label{sec:conditioned}

The process-tensor description acquires its black-hole interpretation once the
radiation is divided according to the information available to an asymptotic
observer. At retarded time $t$, let $R$ denote the radiation collected before
the late experiment and let $L=\{b_1,\ldots,b_n\}$ denote the ordered late bins
on which the observer applies instruments from $\mathcal{I}$. The degrees of
freedom that transmit correlations between the late emissions are represented
by an interior memory $M$. In the finite-dimensional formulation, $M$ is a
tensor factor; in gravitational language, it represents the predictive content
of the island algebra $\mathcal{A}_I$. We consider evaporation processes for
which $M$ is a sufficient memory of the late radiation,
\begin{equation}
  \mathcal{N}_M(L)\le\delta,
  \label{eq:interior-sufficiency}
\end{equation}
so that supplying $M$ screens the earlier late bins from the future up to
accuracy $\delta$. The Page transition can then be expressed as a change in
whether the early radiation supplies the same predictive information.

To compare the predictive information available with and without $R$, the
observer has two operational descriptions of the same evaporation process.
A late-only experiment traces out the early register and probes the marginal
comb $\Upsilon_L=\mathrm{Tr}_R\,\Upsilon_{RL}$, where $\Upsilon_{RL}$ carries
$R$ on an initial side leg and the bins of $L$ on the subsequent intervention
slots. An assisted experiment retains $R$ and probes $\Upsilon_{L|R}$. Its
tester may process $R$ before the first late intervention and carry the
resulting classical or quantum memory through the experiment, while the
interventions on the radiation bins continue to belong to $\mathcal{I}$.

The two access conditions have corresponding Markovian reference classes. The
late-only reference is the Markovian reference set $\mathrm{CPD}$ fixed in
Section~\ref{sec:framework}. With $R$ available as a static side register, the
reference becomes $\mathrm{CPD}_R$, the assisted combs whose late legs are
Markovian once that side information is supplied. Their temporal memories are
\begin{equation}
  \mathcal{N}_\varnothing
  =\min_{\Lambda\in\mathrm{CPD}}
   d_{\mathrm{op}}\!\left(\Upsilon_L,\Lambda\right),
  \qquad
  \mathcal{N}_R
  =\min_{\Lambda\in\mathrm{CPD}_R}
   d_{\mathrm{op}}\!\left(\Upsilon_{L|R},\Lambda\right),
  \label{eq:conditioning}
\end{equation}
which specialize the operational memory measure of~\eqref{eq:Ndef}. Both
quantities refer to the same late dynamics and the same instrument family; they
differ only in the predictive information supplied at the initial leg.

The difference between these two memories defines the island order parameter,
\begin{equation}
  \chi_{\mathrm{island}}(t)
  =\mathcal{N}_\varnothing(t)-\mathcal{N}_R(t).
  \label{eq:chi}
\end{equation}
This quantity is the amount of late-radiation memory screened by access to the
early radiation. For an arbitrary side register, its value records the balance
between two operational effects: side information can expose correlations to
the tester, and it can supply the common cause that renders those correlations
conditionally Markovian. The evaporation setting selects the second behavior
when $R$ reconstructs the predictive influence of $M$.

The underlying structure is the screening implication
\begin{equation}
  \big(M\ \text{Markovianizes}\ L\big)
  \quad\text{and}\quad
  \big(R\xrightarrow{\ \mathcal{R}_{R\to M}\ }M\big)
  \quad\Longrightarrow\quad
  \big(R\ \text{Markovianizes}\ L\big),
  \label{eq:screening}
\end{equation}
where $\mathcal{R}_{R\to M}$ is a recovery or simulation map. Reconstruction is
understood on a code subspace $C$, with error $\epsilon(C)$. Equation
\eqref{eq:interior-sufficiency} supplies the screening accuracy of the
interior, while $\epsilon(C)$ measures how faithfully $R$ reproduces its
predictive action. As the recovery quality improves across the Page regime,
$\mathcal{N}_R$ approaches the residual scale set by $\delta$ and
$\epsilon(C)$, whereas $\mathcal{N}_\varnothing$ continues to measure the
memory visible to a late-only observer. The characteristic transition therefore
consists of a persistent $\mathcal{N}_\varnothing$, a decreasing
$\mathcal{N}_R$, and a rise of $\chi_{\mathrm{island}}$ from $o(1)$ to
$\Theta(1)$.

The same reorganization can be stated through the amount of predictive memory
that must be supplied in addition to $R$. At tolerance $\epsilon$, define
\begin{equation}
  m(L|R)\equiv m_\epsilon(L|R)
  =\inf_{M'}\big\{\,\log\dim M'\ :\
    \mathcal{N}_{RM'}(L)\le\epsilon\,\big\},
  \label{eq:mcost}
\end{equation}
where the infimum ranges over finite-dimensional auxiliary registers admitted
by an extension of the late process. Restricting $M'$ to sufficient memories of
a Markov dilation gives $m_\epsilon$ its direct interpretation as the
predictive memory cost; the minimal such memory, written $M_{\min}$, is the
smallest register through which the future of the late process depends on its
past. When $R$ carries little information about the interior, the sufficient
auxiliary register has the scale of this minimal memory,
$m_\epsilon(L|R)\simeq\log\dim M_{\min}$, with
$\log\dim M_{\min}\le\log\dim M$ and equality when the interior is itself the
minimal memory, $M_{\min}=M$. When $R$ simulates the influence of $M_{\min}$
within the chosen tolerance, the auxiliary cost approaches zero. Thus
$\chi_{\mathrm{island}}$ measures the change through operational
distinguishability, while $m_\epsilon(L|R)$ measures the same change through
the size of a sufficient memory.

A multi-time experiment also carries a resolution $K$, the largest number of
intervention slots spanned by its testers. Let $d_{\mathrm{op}}^{(K)}$ denote
the operational distance restricted to this tester class. For
$Z\in\{\varnothing,R\}$, write
\begin{equation}
  \mathcal{N}^{(K)}_Z
  =\min_{\Lambda\in\mathrm{CPD}_Z}
   d_{\mathrm{op}}^{(K)}\!\left(\Upsilon_{L|Z},\Lambda\right),
  \qquad
  \Upsilon_{L|\varnothing}\equiv\Upsilon_L,
  \qquad
  \mathrm{CPD}_\varnothing\equiv\mathrm{CPD}.
  \label{eq:orderK}
\end{equation}
The tester classes are nested, so each memory estimate $\mathcal{N}^{(K)}_Z$
increases with $K$ and converges to $\mathcal{N}_Z$. The saturation order
$k_\ast$ is the smallest resolution at which the value remains fixed within the
chosen tolerance. The order-resolved island observable
\begin{equation}
  \chi_{\mathrm{island}}^{(K)}
  =\mathcal{N}^{(K)}_\varnothing-\mathcal{N}^{(K)}_R
  \label{eq:chiK}
\end{equation}
converges to~\eqref{eq:chi} and becomes constant once both constituent memories
have saturated. Its first resolved order identifies the temporal scale at which
the early radiation begins to screen the hidden memory; that order can lie
below the saturation order $k_\ast$ and coincides with it when the memory
occupies a single temporal order. Likewise, $m^{(K)}_\epsilon(L|R)$ is obtained
from~\eqref{eq:mcost} by replacing $\mathcal{N}_{RM'}$ with
$\mathcal{N}^{(K)}_{RM'}$; increasing $K$ imposes a finer predictive
requirement on the auxiliary register.

The minimal instance of the transition sits at three time steps. A passive
three-bin experiment admits the conditional-mutual-information pair
\begin{equation}
  W^{(3)}_\varnothing
  =I\!\left(b_1:b_3\,\middle|\,b_2\right),
  \qquad
  W^{(3)}_R
  =I\!\left(b_1:b_3\,\middle|\,b_2 R\right).
  \label{eq:witness}
\end{equation}
These quantities test whether the middle bin, together with the available side
information, screens the first bin from the third. For passive
measure-and-prepare statistics with a classical conditioning register, the
relative-entropy distance to the nearest conditional product evaluates exactly
to these conditional mutual informations, and for general instrument families
they lower-bound the corresponding operational memories in the
type-II-exponent reading. Their difference,
\begin{equation}
  \chi_{\mathrm{island}}^{\mathrm{wit},(3)}
  =W^{(3)}_\varnothing-W^{(3)}_R,
  \label{eq:chiwitness}
\end{equation}
isolates the three-time correlation screened by $R$. A process can hold every
one- and two-time reduced process, as a multilinear functional on instrument
sequences, at a Markovian reference while carrying an order-one value of
$W^{(3)}_\varnothing$; reconstruction of the common-cause memory then
suppresses $W^{(3)}_R$, and the transition is resolved at $k_\ast=3$ with every
lower-order response pinned to the reference. This is the minimal form of the
Markov-order transition.

\section{Main results}
\label{sec:results}

The results below concern the evaporation processes of
Section~\ref{sec:conditioned} satisfying the interior-sufficiency condition
\eqref{eq:interior-sufficiency}, with additional assumptions stated in each
result. We use the process-tensor and quantum-Markov-order framework of
Refs.~\cite{Pollock,Taranto}. The discrimination interpretation follows
Refs.~\cite{CDP,GutoskiWatrous}, and the recovery estimates use
Refs.~\cite{FawziRenner,Petz,JungeUniversal}.

\begin{theorem}[Multi-time blindness]\label{thm:blindness}
There exist a one-parameter family of evaporation processes $\{\Upsilon_t\}$
and an integer $k_\ast\ge 3$ such that, across $t_{\mathrm{Page}}$,
\begin{enumerate}
  \item every single-bin marginal is thermal and unchanged,
        $\rho_{b_i}=\tau_\beta$;
  \item every marginal of order at most $k_\ast-1$ coincides with a fixed
        CP-divisible reference $\Lambda$, before and after the transition;
  \item the order-$k_\ast$ estimator $\chi_{\mathrm{island}}^{(k_\ast)}$ is
        $o(1)$ before $t_{\mathrm{Page}}$ and $\Theta(1)$ after.
\end{enumerate}
For this family $\chi_{\mathrm{island}}^{(K)}=0$ for $K<k_\ast$, and the memory
occupies the single order $k_\ast$, so
$\chi_{\mathrm{island}}^{(k_\ast)}=\chi_{\mathrm{island}}$. For this family,
probes spanning fewer than $k_\ast$ slots cannot distinguish the process from
$\Lambda$ and therefore do not detect the transition.
The minimal construction has $k_\ast=3$.
\end{theorem}

\emph{Proof sketch.} The family is a marginal-matching construction whose
marginals of order at most $k_\ast-1$ coincide with those of a CP-divisible
$\Lambda$, while its order-$k_\ast$ conditional structure carries the
correlation that $M$ mediates among the late bins. Item~1 holds because each
bin is prepared in $\tau_\beta$. Item~2 holds by the marginal matching, which
fixes $\chi_{\mathrm{island}}^{(K)}=0$ for $K<k_\ast$. Item~3 holds in two
steps: before $t_{\mathrm{Page}}$ the early register carries $o(1)$ correlation
with $M$, so conditioning preserves the comb,
$\mathcal{N}_R^{(k_\ast)}\simeq\mathcal{N}_\varnothing^{(k_\ast)}$, and
$\chi_{\mathrm{island}}^{(k_\ast)}=o(1)$; after $t_{\mathrm{Page}}$,
Theorem~\ref{thm:bridge} at order $k_\ast$ gives
$\mathcal{N}_R^{(k_\ast)}\to 0$ while
$\mathcal{N}_\varnothing^{(k_\ast)}=\Theta(1)$, so
$\chi_{\mathrm{island}}^{(k_\ast)}=\Theta(1)$. Appendix~\ref{app:parity}
constructs the $k_\ast=3$ instance, a shared-parity comb, and
Appendix~\ref{app:endpoints} computes its endpoint values.

The construction realizes the superactivation of memory by combining time
steps reported in Ref.~\cite{priorpaper}, here at a definite order $k_\ast$
with the lower-order marginals fixed to a Markovian reference.

\begin{proposition}[Discrimination reading]\label{prop:opmeaning}
The tester distinguishability of two combs is
\begin{equation}
  d_{\mathrm{op}}(\Upsilon,\Lambda)=\sup_{T\in\mathcal{I}}\
    D\big(T[\Upsilon]\,\big\|\,T[\Lambda]\big),
  \label{eq:dopexplicit}
\end{equation}
the optimal performance of a tester from $\mathcal{I}$ in discriminating
$\Upsilon$ from $\Lambda$, with $D$ the base divergence whose figure of merit
is the optimal success probability or the type-II error exponent. The memory
$\mathcal{N}[\Upsilon]=\min_{\Lambda\in\mathrm{CPD}}d_{\mathrm{op}}(\Upsilon,\Lambda)$
is then the optimal distinguishability of $\Upsilon$ from its nearest Markovian
comb, and the order parameter
$\chi_{\mathrm{island}}=\mathcal{N}_\varnothing-\mathcal{N}_R$ is the
difference of two such optima, the distinguishability from the Markovian set
that conditioning on $R$ removes, equivalently the power of $R$ to Markovianize
the late process.
\end{proposition}

\emph{Proof sketch.} Equation~\eqref{eq:dopexplicit} is the tester-based
divergence of Section~\ref{sec:framework}, so the supremum is the optimal
distinguishing performance against a fixed Markovian comb and the minimum over
$\mathrm{CPD}$ selects the nearest one. The order parameter compares the two
optima, for $\Upsilon_L$ and for $\Upsilon_{L|R}$; each nearest comb is a
separate optimizer, so the difference of the two optima is the precise
statement. Appendix~\ref{app:dual} records the figure of merit and this
difference.

This discrimination interpretation concerns the observer's ability to
distinguish processes; it does not imply reconstruction of the interior state.

\begin{theorem}[Forward bridge]\label{thm:bridge}
Assume:
\begin{enumerate}[label=\textup{(A\arabic*)},leftmargin=*]
  \item the interior Markovianizes the late process,
  $\mathcal{N}_M(L)\le\delta$;
  \item a recovery map $\mathcal{R}_{R\to M}$ reconstructs the interior from
  the early radiation on a code subspace $C$ with infidelity at most
  $\epsilon(C)$,
  \begin{equation}
  \begin{aligned}
    1-F\big(&\rho_{ML},
    (\mathcal{R}_{R\to M}\otimes\mathrm{id}_L)(\rho_{RL})\big)\\
    &\le\epsilon(C)
    \qquad\text{for all } \rho\in C.
  \end{aligned}
    \label{eq:A2}
  \end{equation}
\end{enumerate}
Then for states in $C$,
\begin{equation}
  \mathcal{N}_R(L)\big|_C\ \le\ \delta+\omega_{\mathcal{I},n}\big(\epsilon(C)\big),
  \label{eq:T3bound}
\end{equation}
where $\omega_{\mathcal{I},n}$ is the continuity modulus of $d_{\mathrm{op}}$
for the instrument family $\mathcal{I}$ on $n$ bins,
$\omega_{\mathcal{I},n}(0)=0$. For the infidelity of~\eqref{eq:A2},
$\omega_{\mathcal{I},n}(\epsilon)=c_{\mathcal{I},n}\sqrt{\epsilon}$ with
$c_{\mathcal{I},n}=O(n)$; equivalently
$\chi_{\mathrm{island}}\big|_C\ge\mathcal{N}_\varnothing-\delta
 -c_{\mathcal{I},n}\sqrt{\epsilon(C)}$.
\end{theorem}

\emph{Proof sketch.} Compose the recovery $\mathcal{R}_{R\to M}$ with the
$M$-conditioned process. By A2 and the Fuchs and van de Graaf inequalities the
recovered conditioning state differs from the true one in trace distance by at
most $\sqrt{2\,\epsilon(C)}$, and data processing for $d_{\mathrm{op}}$
propagates this through the $n$ late bins, so replacing $M$ by its
reconstruction costs $c_{\mathcal{I},n}\sqrt{\epsilon(C)}$ with
$c_{\mathcal{I},n}=O(n)$. Assumption A1 supplies $\delta$, and the triangle
inequality gives~\eqref{eq:T3bound}. A trace-norm or diamond-norm definition of
the reconstruction error replaces the modulus by the linear
$O(n\,\epsilon(C))$. Appendix~\ref{app:constant} fixes the constant.

The bound applies on the code subspace $C$ where the stated reconstruction
error $\epsilon(C)$ is controlled; it does not require recovery outside $C$.

\begin{theorem}[Partial converse]\label{thm:converse}
Suppose the late process has a minimal Markov dilation $M_{\min}$, unique up to
predictive equivalence \textup{(B1)}; every future statistic accessible to
$\mathcal{I}$ factors through $M_{\min}$, with no independent side channels
\textup{(B2)}; and $R$ Markovianizes the late process in the tester norm,
$\mathcal{N}_R(L)\le\epsilon$ \textup{(B3)}. Then there is a simulator
$\mathcal{R}_{R\to M_{\min}}$ reproducing the influence of $M_{\min}$ on the
future late-radiation statistics,
\begin{equation}
  \big\|\,\Phi_{\mathrm{fut}\mid M_{\min}}\circ\mathcal{R}_{R\to M_{\min}}
    -\Phi_{\mathrm{fut}\mid R}\,\big\|_{\mathcal{I}_{\mathrm{fut}}}\ \le\
  g(\epsilon),
  \qquad g(\epsilon)=O(\sqrt{\epsilon}),
  \label{eq:T4bound}
\end{equation}
with $\Phi_{\mathrm{fut}\mid\cdot}$ the predictive channel to the future late
bins and $\|\cdot\|_{\mathcal{I}_{\mathrm{fut}}}$ the distinguishability under
future testers from $\mathcal{I}$.
\end{theorem}

\emph{Proof sketch.} By B1 and B2 every future-affecting degree of freedom
factors through $M_{\min}$, so reproducing its action on the future reproduces
the full late-time memory. By B3, conditioning on $R$ brings the late process
within $\epsilon$ of Markovian in the tester norm, so $R$ is an approximate
sufficient statistic for $M_{\min}$, and approximate sufficiency yields a
channel $\mathcal{R}_{R\to M_{\min}}$ satisfying~\eqref{eq:T4bound}.
Appendix~\ref{app:simulator} gives the construction and the square-root
estimate.

Under B1--B3, the simulator reproduces the action of $M_{\min}$ on future
statistics, rather than necessarily reconstructing the microscopic interior
state.

\begin{proposition}[Memory-cost collapse]\label{prop:memcost}
The conditional memory cost is the minimal sufficient memory before the Page
time and collapses after,
\begin{equation}
  m_\epsilon(L\mid R)\simeq\log\dim M_{\min}\ \ (t<t_{\mathrm{Page}}),
  \qquad
  m_\epsilon(L\mid R)=0\ \ (t>t_{\mathrm{Page}}),
  \label{eq:mcollapse}
\end{equation}
the post-Page value holding once $\epsilon(C)$ is below the tolerance
$\epsilon$, with $m_0(L\mid R)=0$ in the exact-recovery limit. In the faithful
instance, where $M_{\min}=M$ with no coarse-graining,
$m_\epsilon(L\mid R)\simeq\log\dim M$ before $t_{\mathrm{Page}}$ and the
trivial register suffices after.
\end{proposition}

The screening register of
the $k_\ast=3$ instance is identified in Appendix~\ref{app:screening}.

\section{Minimal evaporation-comb model}
\label{sec:model}

A single finite-dimensional family realizes the two ingredients of the
transition: screening of a common-cause memory by the early radiation and
localization of that memory at a definite temporal order. The first ingredient
is transparent in a two-bin screening core, where all quantities are available
in closed form. The second is supplied by a three-bin parity extension whose
one- and two-time marginals coincide with a fixed Markovian reference. Both
belong to the evaporation-comb family
\begin{equation}
  \mathsf{E}_{\lambda,C}
  =\Big(
      \rho_{AM},\{\mathcal{E}_i\}_{i=1}^{n},
      \mathcal{R}_{R\to\widehat M}(\lambda),\mathcal{I},C
    \Big),
  \label{eq:model-family}
\end{equation}
where $M$ is the interior memory, $A$ is an external reference that keeps it
mixed, $\mathcal{E}_i$ emits the late bin $b_i$, and
$\mathcal{R}_{R\to\widehat M}$ decodes from the early radiation a predictive
memory $\widehat M$. The coordinate $\lambda\in[0,1]$ orders the recovery path:
$\lambda=0$ supplies a trivial decoder, while $\lambda=1$ supplies an exact
simulator of the part of $M$ that controls the future late statistics. Its
relation to the reconstruction error is fixed below by the growth of the early
register.

The causal structure of the family is explicit in the dilation of each
emission step,
\begin{equation}
  U_i:\ M_{i-1}\otimes e_i
  \longrightarrow M_i\otimes b_i\otimes f_i,
  \qquad
  \mathcal{E}_i(\rho_{M_{i-1}})
  =\mathrm{Tr}_{f_i}\!\left[
       U_i(\rho_{M_{i-1}}\otimes\tau_{e_i})U_i^\dagger
    \right],
  \label{eq:model-emission}
\end{equation}
with a fresh seed $e_i$, a discarded environment $f_i$, and the predictive
memory $M_i$ carried to the next step. Once the value of the carried memory is
supplied, the remaining randomness is fresh at every slot. The resulting comb
therefore obeys
\begin{equation}
  \mathcal{N}_M(L)\leq\delta_{\mathrm{model}},
  \label{eq:model-sufficiency}
\end{equation}
with $\delta_{\mathrm{model}}=0$ for the measure-and-prepare constructions
below. Controlled emission gives the solvable classical core. Coherent
partial-SWAP and scrambling unitaries provide deformations of the same causal
structure whose resolved memory depends on the instrument family.

For the two-bin screening core, let $d_M=2$ and initialize $A$ and $M$ in a
Bell pair. The predictive record is
the computational-basis value $z\in\{0,1\}$. Controlled emission writes this
value into two late bins, while the early register carries one of two pure
states whose overlap decreases along the recovery path:
\begin{align}
  |\Psi_\lambda\rangle
  &=\frac{1}{\sqrt{2}}
    \sum_{z=0}^{1}
    |z\rangle_A|z\rangle_M|r_z(\lambda)\rangle_R
    |z\rangle_{b_1}|z\rangle_{b_2},
  \label{eq:screening-state}\\
  |r_0(\lambda)\rangle&=|0\rangle,
  \qquad
  |r_1(\lambda)\rangle
  =\cos\!\left(\frac{\pi\lambda}{2}\right)|0\rangle
   +\sin\!\left(\frac{\pi\lambda}{2}\right)|1\rangle.
  \label{eq:recovery-states}
\end{align}
The late-only observer sees the same perfectly correlated pair at every
$\lambda$; the assisted observer distinguishes the two values of the
common-cause record with an error that decreases along the recovery path.
Writing
$s_\lambda=|\langle r_0(\lambda)|r_1(\lambda)\rangle|
=\cos(\pi\lambda/2)$, the conditional-mutual-information witness is
\begin{align}
  \mathcal{N}^{\mathrm{wit}}_\varnothing
  &=I(b_1:b_2)=1,
  \nonumber\\
  \mathcal{N}^{\mathrm{wit}}_R(\lambda)
  &=I(b_1:b_2\mid R)
    =1-h_2\!\left(\frac{1+s_\lambda}{2}\right),
  \label{eq:screening-closed}\\
  \chi^{\mathrm{wit}}_{\mathrm{island}}(\lambda)
  &=h_2\!\left(\frac{1+s_\lambda}{2}\right),
  \label{eq:screening-chi}
\end{align}
where $h_2$ is the binary entropy in bits. At $\lambda=0$ the two states of
$R$ coincide, so conditioning preserves the full one-bit memory. At
$\lambda=1$ they are orthogonal, $R$ determines the predictive record, and the
conditioned memory vanishes. Every single-bin marginal remains at the thermal
value of the model, the maximally mixed state with $S(b_i)=1$, throughout the
interpolation.

An erasure decoder gives the degradation-ordered version of the same path. It
reveals the record with probability $\lambda$ and returns an erasure flag
otherwise, so
\begin{equation}
  \mathcal{N}^{\mathrm{wit}}_R(\lambda)=1-\lambda,
  \qquad
  \chi^{\mathrm{wit}}_{\mathrm{island}}(\lambda)=\lambda.
  \label{eq:erasure-path}
\end{equation}
The decoder at smaller $\lambda$ is obtained from the larger-$\lambda$ decoder
by additional erasure noise. Data processing therefore fixes the monotone
ordering of the assisted witness along the path. The coherent path in
\eqref{eq:screening-closed} supplies a smooth interpolation with the same
endpoints.

The screening core isolates the conditioning mechanism, while the minimal
Markov-order separation requires three late slots. Let $X$ and $Y$ be
independent uniform bits. A causal one-bit dilation retains $X$ after the first
emission, draws $Y$ as the fresh seed at the second step, and updates the
carried memory to $X\oplus Y$ before the third emission:
\begin{equation}
  b_1=X,
  \qquad
  b_2=Y,
  \qquad
  b_3=X\oplus Y.
  \label{eq:parity-process}
\end{equation}
Equivalently,
\begin{equation}
  p(b_1,b_2,b_3)
  =\frac{1}{4}\,
    \mathbf{1}\!\left[b_1\oplus b_2\oplus b_3=0\right].
  \label{eq:parity-distribution}
\end{equation}
Every one-bin marginal and every two-bin marginal is the uniform product, so
all responses below order three coincide with the fixed measure-and-prepare
Markovian reference that prepares an independent fair bit at each slot. The
three-time witness is nevertheless
\begin{equation}
  W^{(3)}_\varnothing
  =I(b_1:b_3\mid b_2)=1\ \mathrm{bit}.
  \label{eq:parity-witness}
\end{equation}
Let the early register encode the retained predictive bit $X$ through the
states~\eqref{eq:recovery-states}. Since $b_2=Y$ is already the conditioning
bin,
\begin{align}
  W^{(3)}_R(\lambda)
  &=I(X:X\oplus Y\mid YR)
    =H(X\mid R)
    =1-h_2\!\left(\frac{1+s_\lambda}{2}\right),
  \label{eq:parity-conditioned}\\
  \chi^{\mathrm{wit},(3)}_{\mathrm{island}}(\lambda)
  &=h_2\!\left(\frac{1+s_\lambda}{2}\right).
  \label{eq:parity-order-parameter}
\end{align}
The one- and two-time combs are independent of $\lambda$, while the order-three
memory is screened as the early radiation recovers $X$. This realizes
Theorem~\ref{thm:blindness} with $k_\ast=3$. It also identifies the minimal
sufficient memory: one retained bit Markovianizes the late process before the
Page point, and the trivial register suffices after exact recovery. The full
interior implementation may contain further seeds and spectator degrees of
freedom, but the conditional memory cost is governed by this minimal
predictive register $M_{\min}$.

The same construction places the memory at any prescribed order $K_0\geq3$.
For independent fair bits $X_1,\ldots,X_{K_0-1}$, take
\begin{equation}
  b_i=X_i\quad(1\leq i<K_0),
  \qquad
  b_{K_0}=\bigoplus_{i=1}^{K_0-1}X_i.
  \label{eq:higher-parity}
\end{equation}
Every proper marginal is uniform, while the full $K_0$-time parity is fixed.
A running parity bit supplies a minimal causal memory, and the order-resolved
observable first becomes nonzero at $K=K_0$.

To relate these recovery paths to evaporation time, we use the Hayden--Preskill
decoupling picture~\cite{HaydenPreskill}. The parameter $\lambda$ is a
coordinate on the decoder family rather than an independent coupling.
Its retarded-time dependence is determined by
the amount of early radiation available relative to the predictive memory and
the code subspace. Define the overload
\begin{equation}
  \mu(t,C)
  =\log_2 d_{R,\mathrm{eff}}(t)
   -\log_2 d_{M_{\min}}
   -\log_2\dim C,
  \label{eq:overload}
\end{equation}
where $d_{R,\mathrm{eff}}(t)$ is the effective dimension of the collected
radiation participating in the decoder. The finite-dimensional scans use the
calibration
\begin{equation}
  \epsilon_C(t)
  =2^{-\mu_+(t,C)/2},
  \qquad
  \mu_+(t,C)=\max\{\mu(t,C),0\},
  \qquad
  \lambda(t,C)=\ell\!\left(\epsilon_C(t)\right),
  \label{eq:recovery-calibration}
\end{equation}
with $\ell$ monotone decreasing, $\ell(1)=0$, and $\ell(0)=1$. For the erasure
path, $\ell(\epsilon)=1-\epsilon$; for the coherent path, $\ell$ fixes the
overlap $s_\lambda$ and hence the optimal binary decoding error. Equation
\eqref{eq:recovery-calibration} is the model calibration of the decoupling
threshold: increasing $d_{R,\mathrm{eff}}$ improves recovery, while increasing
$d_{M_{\min}}$ or $\dim C$ delays it.

Combining interior sufficiency with the reconstruction error gives the
post-Page estimate
\begin{equation}
  \chi_{\mathrm{island}}(t)\big|_C
  \geq
  \mathcal{N}_\varnothing(t)
  -\delta_{\mathrm{model}}
  -c_{\mathcal{I},n}\sqrt{\epsilon_C(t)},
  \label{eq:postpage}
\end{equation}
which is the specialization of Theorem~\ref{thm:bridge} to
$\mathsf{E}_{\lambda,C}$. A larger code subspace shifts the decoding threshold
to later retarded time and leaves a larger finite-error remainder. Varying
$d_{M_{\min}}$ traces the complementary resource statement: before recovery
the auxiliary cost is $\log_2 d_{M_{\min}}$, while after recovery it falls to
zero within the selected tolerance, as in Proposition~\ref{prop:memcost}.

\begin{figure}[t]
  \centering
  \includegraphics[width=\textwidth]{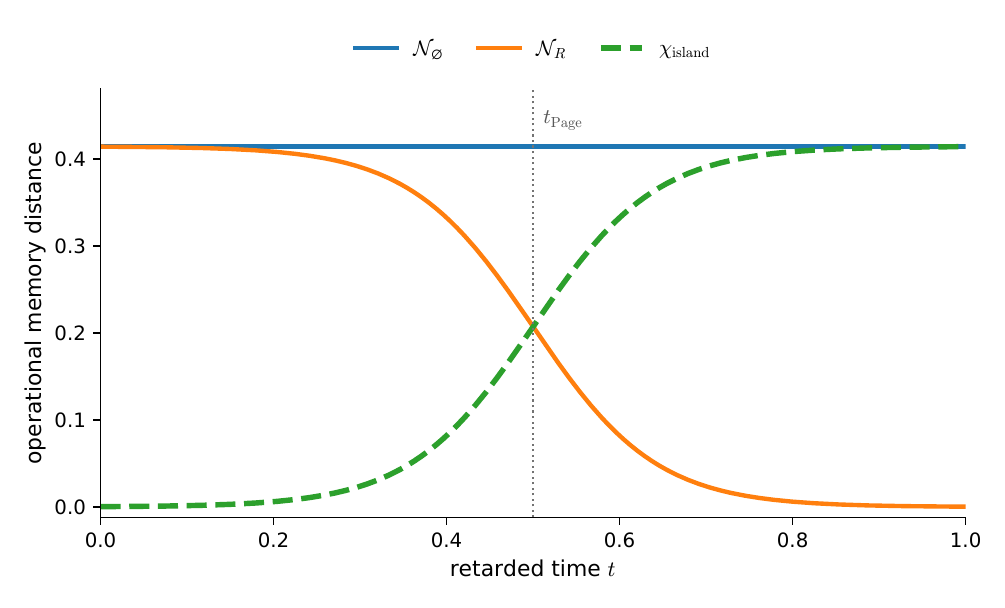}
  \caption{Operational Page-time transition in the minimal evaporation comb.
  Retarded time enters through the recovery coordinate $\lambda(t,C)$, and the
  curves report the operational measure, with endpoints fixed by the
  closed-form witnesses. The unconditioned late-radiation memory
  $\mathcal{N}_\varnothing$ remains finite across the transition, while the
  memory $\mathcal{N}_R$ conditioned on the early radiation decreases as the
  decoder becomes effective. Their difference
  $\chi_{\mathrm{island}}=\mathcal{N}_\varnothing-\mathcal{N}_R$ consequently
  rises from $o(1)$ to $\Theta(1)$. The single-bin entropy remains fixed at its
  thermal value throughout the sweep. For the parity extension the same curves
  are first resolved at order $k_\ast=3$, while every one- and two-time reduced
  process continues to coincide with the Markovian reference.}
  \label{fig:central}
\end{figure}

Figure~\ref{fig:central} shows the operational memory for the family
$\mathsf{E}_{\lambda,C}$. Equations~\eqref{eq:screening-closed}
and~\eqref{eq:parity-conditioned} give computable witnesses and fix the
endpoint values. Before decoding, the early register contains too
little information about $M_{\min}$ to screen the late bins, so
$\mathcal{N}_R\simeq\mathcal{N}_\varnothing$ and
$\chi_{\mathrm{island}}\simeq0$; once the overload becomes positive, the
recovered predictive record screens the common cause and drives
$\mathcal{N}_R$ toward the residual scale
$\delta_{\mathrm{model}}+c_{\mathcal I,n}\sqrt{\epsilon_C}$.
Hayden--Preskill decoupling determines when the decoder becomes effective,
and the code subspace determines its recovery accuracy.

The single-bin entropy remains constant throughout the scan. In the parity
realization the pair processes are also fixed, so detecting the transition
requires probes spanning at least three slots.

\section{Numerical checks and robustness}
\label{sec:numerics}

This section establishes the robustness of the transition on the single family
$\mathsf{E}_{\lambda,C}$ of Section~\ref{sec:model}. Figure~\ref{fig:robustness} collects
four panels: three axes of robustness, the instrument family, the code subspace, and the
diagnostic order, together with the collapse of the memory cost. All plotted quantities are
the trace-norm distance $\mathcal{N}$, which for the passive families, those in which the
observer collects and measures each bin, is the total variation of the outcome statistics,
with success probability $(1+\mathcal{N})/2$. Each curve carries one of three provenances,
stated in the caption: closed form, where the minimization over the Markovian reference
solves analytically; grid-verified, where the numerical minimization is checked under grid
refinement; and bracketed, where the value is reported between a certified lower and upper
bound. The grid-verified values do not carry a rigorous discretization-error enclosure.

Finite computations require stated tolerances, and the convergence protocol fixes them once.
The numerical floor $\epsilon_{\mathrm{num}}$ is the accuracy of the searches and solvers.
The truncation residual
\begin{equation}
  \epsilon_{\mathrm{trunc}}(K)
  =\max_j|\mathcal{N}^{(K+j+1)}-\mathcal{N}^{(K+j)}|
  \label{eq:order-residual}
\end{equation}
measures how much the estimate still moves when the probe order is raised, and the
saturation order $k_\ast$ is the first $K$ past which this residual stays below tolerance.
The memory-cost tolerance is chosen to satisfy
\begin{equation}
  \epsilon_\ast>
  \max\{\epsilon(C),\epsilon_{\mathrm{num}},\epsilon_{\mathrm{trunc}}(k_\ast)\},
  \label{eq:memory-cost-tolerance}
\end{equation}
so that the
reported collapse is resolved above every error source.

\begin{figure}[t]
\centering
\includegraphics[width=\textwidth]{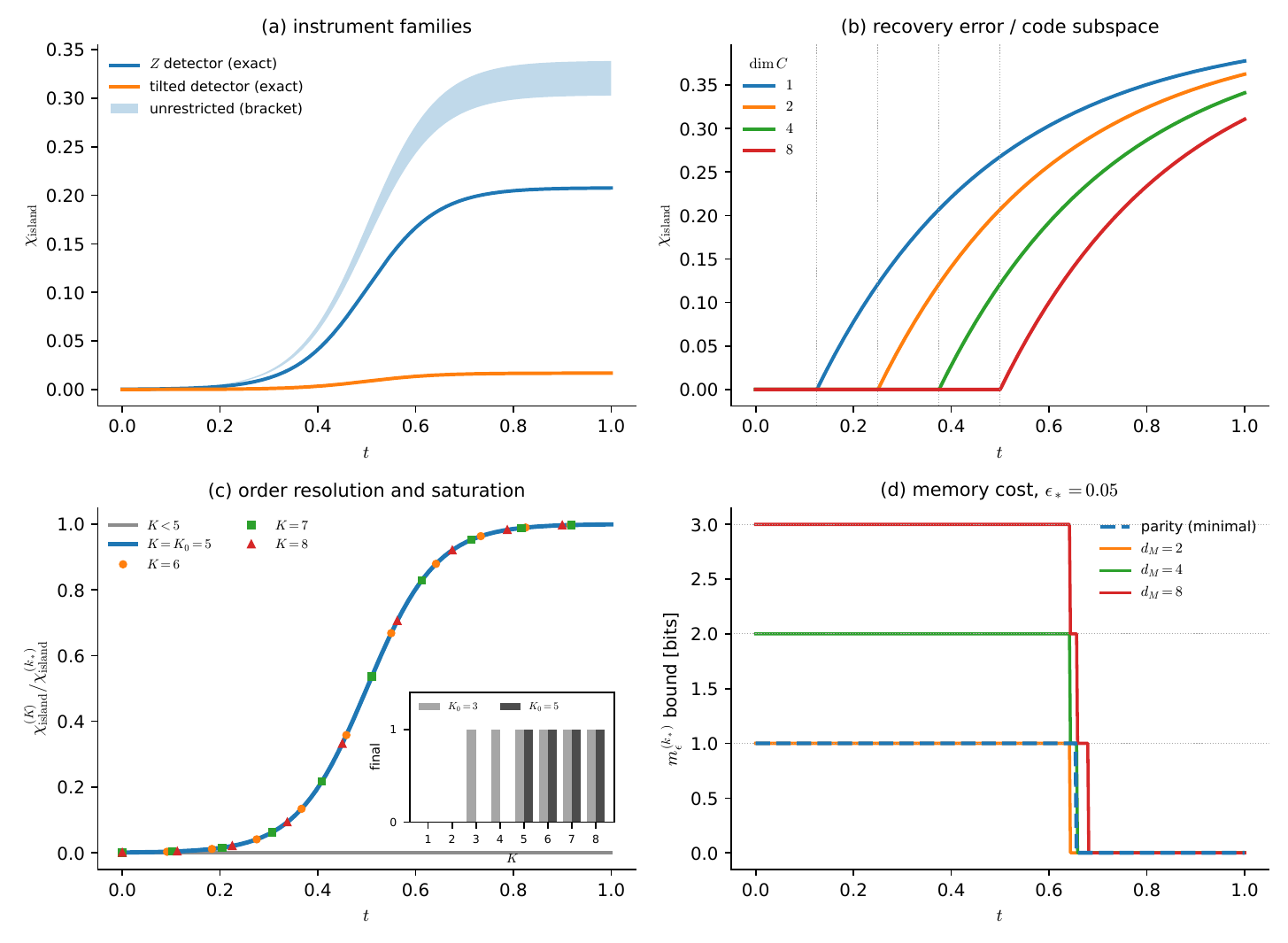}
\caption{Numerical robustness of the Page-time transition in the minimal evaporation comb,
along three axes of robustness and the memory-cost collapse.\ (a) Instrument-family dependence,
on the coherent emission of the production model: the resolved memory $\mathcal{N}_\varnothing$
spans an order of magnitude across a tilted detector, the computational detector, and the
unrestricted family shown as a bracket, while $\chi_{\mathrm{island}}$ rises from $o(1)$ to
$\Theta(1)$ across $t_{\mathrm{Page}}$ for every family.\ (b) Recovery error and code subspace:
increasing $\dim C$ shifts the Page point later through the overload~\eqref{eq:overload} and
shallows the post-Page plateau, bounded below by~\eqref{eq:postpage}.\ (c) Order resolution and
saturation: for the order-$K_0$ parity family the normalized estimator
$\chi_{\mathrm{island}}^{(K)}/\chi_{\mathrm{island}}^{(k_\ast)}$ vanishes for $K<K_0$ and equals
one for $K\ge K_0$, saturating at $k_\ast=K_0$, here $K_0=5$; the inset shows the final-time
value against $K$ for the minimal instance $K_0=3$ and for $K_0=5$.\ (d) Conditional
memory-cost bound $m_\epsilon^{(k_\ast)}$: for the parity family the minimal predictive cost
collapses from one bit to zero, while the common-cause family collapses from $\log_2 d_M$ to
zero for $d_M\in\{2,4,8\}$. All curves are closed form, grid-verified, or bracketed, and the
common-cause memory-cost curves are upper bounds.}
\label{fig:robustness}
\end{figure}

Panel (a) answers the first question: does the transition depend on the probe? Memory is
defined relative to the instrument family, the operational stance of Section~\ref{sec:framework},
so the physical content is the change of this relative quantity across $t_{\mathrm{Page}}$,
and the transition must survive every reasonable family. The panel runs on the coherent
emission of the production model, the setting in which the choice of instrument matters: the
bins carry coherences, and a detector resolves more or less memory according to its alignment
with the emission basis. Three families are compared. The computational detector reads the
emission basis directly. The tilted detector is rotated away from it. The unrestricted value
is the memory available to the best measurement quantum mechanics allows, reported as a
bracket whose lower edge is the best product-basis measurement, a lower bound by data
processing, and whose upper edge is a multistart product-state search, a rigorous upper
bound. The resolved scale of $\mathcal{N}_\varnothing$ spans an order of magnitude, from
$0.017$ for the tilted detector through $0.21$ for the computational detector to the bracket
$[0.30,0.34]$, while every family shows the same rise of $\chi_{\mathrm{island}}$ from $o(1)$
to $\Theta(1)$ across $t_{\mathrm{Page}}$. The scale of the memory is a property of the probe;
its Page-time transition is a property of the process. A control run separates decoding from
mere conditioning: on a scrambling emission with a decoder misaligned from the memory basis,
the conditioned memory stays near the unconditioned value and $\chi_{\mathrm{island}}$ remains
$o(1)$. The rise therefore registers genuine decoding of the interior, the dynamical reading
of Section~\ref{sec:conditioned}.

Panel (b) answers the second question: how does the transition depend on the code subspace on
which reconstruction holds? The recovery quality tracks the growth of the early register
through the Hayden-Preskill overload~\eqref{eq:overload}, the number of early-register qubits
in excess of what the interior and the microstate ensemble require, and this overload sets the
reconstruction error $\epsilon_C(t)$. The panel sweeps $\dim C\in\{1,2,4,8\}$ at fixed $d_M$,
and two effects follow. The Page point shifts later as the code subspace grows, at
$t_{\mathrm{Page}}(C)=(\log_2 d_M+\log_2\dim C)/\rho$ for register growth rate $\rho$: a larger
microstate ensemble requires more early radiation to decode. The post-Page plateau of
$\chi_{\mathrm{island}}$ shallows with $\dim C$, from $0.38$ at $\dim C=1$ to $0.31$ at
$\dim C=8$, and stays above the bound~\eqref{eq:postpage}. The transition survives across the
ensemble, with a magnitude the reconstruction error controls.

Panel (c) answers the third question: is the effect an artifact of probing at a fixed number
of time steps? A transition confined to one diagnostic order would be a property of the probe
rather than the process, so the panel tracks the order-resolved estimators across $K$. These
estimators are nested, since a probe spanning $K$ steps can implement any probe spanning
fewer, so $\chi_{\mathrm{island}}^{(K)}$ increases with $K$ to the full-comb value and settles
at the saturation order $k_\ast$. The test family is the order-$K_0$ parity comb, built so
that every marginal below order $K_0$ matches the Markovian reference, a structure certified
numerically, with exactly one order-$K_0$ marginal carrying the memory. The plotted estimator
is normalized by its saturated value, which removes the overall scale and leaves the panel
testing structure alone: the sequence vanishes for $K<K_0$, equals one for $K\ge K_0$, and the
plateau criterion returns $k_\ast=K_0$. The panel takes $K_0=5$ to display that the saturation
order is read from the computation; the minimal instance $k_\ast=3$ realizes
Theorem~\ref{thm:blindness} and appears in the inset of Figure~\ref{fig:robustness}(c).
The resolved order is a property of the process, and $K$ is the resolution of the experiment.

Panel (d) states the transition as a memory budget. The conditional memory cost
$m_\epsilon(L|R)$ is the smallest side register that, supplied to the tester together with
$R$, renders the late statistics Markovian to tolerance $\epsilon_\ast$, and the panel reports
its saturated bound $m_\epsilon^{(k_\ast)}$ across the sweep. For the parity family the
minimal forward predictive memory is the retained bit $X$: conditioned on $X$, the third bin
is determined by the recent bin $b_2=Y$, so one bit screens the distant past exactly. The
parity cost therefore collapses from one bit to zero once the decoded early radiation brings
the zero-auxiliary residual within tolerance. For the common-cause family the cost collapses
from $\log_2 d_M$ to zero at the decoding threshold, traced for $d_M\in\{2,4,8\}$. These
common-cause values are upper bounds, since a register that achieves the tolerance certifies
the cost from above; their pre-Page level is $\log_2\dim M$, the full interior. The collapse
of this cost is the transition read as the memory an asymptotic observer supplies, the
secondary order parameter of Section~\ref{sec:conditioned}.

The single-bin entropy and the adjacent-pair process tensor hold fixed across
$t_{\mathrm{Page}}$ while $\chi_{\mathrm{island}}$ rises, the comparison set against the order
parameter in the blindness panel, Figure~\ref{fig:robustness}(c). Across the four panels the
transition is robust to the instrument, the code subspace, and the diagnostic order, carried
by the multi-time observable that the one- and two-time diagnostics track only where they
agree with a Markovian reference.

\section{From the QES transition to process Markovianization}
\label{sec:semiclassical}

Semiclassical gravity supplies the physical inputs of the conditioning
results.  The interior memory becomes reconstructible from the collected
radiation when the dominant quantum-extremal-surface saddle includes the
island, and the same reconstruction turns the early radiation into a
predictive register for the future Hawking process.  The resulting process
transition has a characteristic temporal structure.  The four radiation
regions entering the order-three witness exchange saddles at different
retarded times, so the entropy Page time lies inside a finite interval over
which the conditioned multi-time observable reorganizes.  The unconditioned
witness remains fixed in the regime considered here, while the conditioned
witness approaches the island-screened value once the relevant saddles become
compatible and the diary is recoverable from the early radiation.

Consider an evaporating JT black hole coupled to a nongravitating bath,
following Refs.~\cite{AEMM,Hollowood:2020cou,Hollowood:2020couPage}, with an
infalling diary $M$ entangled with an external reference as in
Ref.~\cite{Penington}.  The diary carries the predictive record that correlates
the late bins.  Its ensemble defines the code subspace $C$, while the bath
radiation collected before the late experiment defines the side register $R$.
The correspondence with the finite-dimensional family
$\mathsf E_{\lambda,C}$ is direct: the diary supplies the minimal interior
memory, the geometry determines when a decoder from $R$ becomes available,
and the generalized-entropy gaps determine the recovery coordinate.

For the explicit bookkeeping model we use the Page-regime entropy profile
\begin{equation}
  S_{\mathrm{BH}}(u)=S_0 e^{-u/\tau},
  \qquad
  S_{\mathrm{rad}}(u)=S_0-S_{\mathrm{BH}}(u),
  \qquad
  \tau=\frac{2}{k}.
  \label{eq:aemmexp}
\end{equation}
The second relation fixes an entropy-conserving calibration for the bath and
keeps the parameters directly comparable with the finite-dimensional scans.
The exact Bessel-function boundary reparametrization of the AEMM solution gives
the same Page-regime trajectory and the expected scrambling-scale QES lag to
the accuracy relevant for the figure.  The arbitrary-subset entropies are then
evaluated with the adiabatic interval bookkeeping of
Refs.~\cite{Hollowood:2020cou,Hollowood:2020couPage}.

The interval entropies enter the passive order-three observable,
\begin{equation}
  W_R^{(3)}
  =S(Rb_1b_2)+S(Rb_2b_3)-S(Rb_2)-S(Rb_1b_2b_3),
  \label{eq:fourentropy}
\end{equation}
with the four-region family
\begin{equation}
  \mathcal X_3
  =\big\{Rb_2,\,Rb_1b_2,\,Rb_2b_3,\,Rb_1b_2b_3\big\}.
  \label{eq:region-family}
\end{equation}
For each $X\in\mathcal X_3$, define the generalized-entropy gap
\begin{equation}
  \Delta_X(u)
  =S_{\mathrm{gen}}^{\mathrm{ext}}(X;u)
   -S_{\mathrm{gen}}^{\mathrm{isl}}(X;u).
  \label{eq:gapdef}
\end{equation}
The island saddle dominates for $\Delta_X>0$, and its exchange time
$u_X^\ast$ satisfies $\Delta_X(u_X^\ast)=0$.  Because the four regions have
different late-bin endpoints, their gaps need not vanish simultaneously.

Away from the exchange points, a common island makes the relation to process
Markovianization precise.

\begin{proposition}[Semiclassical reduction of the Markov witness]
\label{prop:semiclassical}
Evaluate the four entropies in~\eqref{eq:fourentropy} with the island
prescription on an adiabatic evaporating background carrying a diary $M$.
Suppose the selected QES saddles form a compatible family: every island
contains the diary, the associated bulk domains admit a common predictive
island $I_\star$, and the geometric terms cancel in the four-entropy
combination up to a defect $\eta_{\mathrm{QES}}^{(3)}$.  Then
\begin{equation}
  W_R^{(3)}
  =I_{\mathrm{bulk}}
    \bigl(b_1:b_3\,\big|\,b_2 R I_\star\bigr)
   +O\!\left(\eta_{\mathrm{QES}}^{(3)}\right).
  \label{eq:semired}
\end{equation}
If the diary is reconstructible from $R$ on the code subspace $C$ with error
$\epsilon(C)$, then
\begin{equation}
  W_R^{(3)}
  \leq
  \delta_{\mathrm{bulk}}^{(3)}
  +c_3\sqrt{\epsilon(C)}
  +\eta_{\mathrm{QES}}^{(3)},
  \label{eq:semibridge}
\end{equation}
where $\delta_{\mathrm{bulk}}^{(3)}$ is the residual order-three bulk memory
after conditioning on the predictive island.
\end{proposition}

\emph{Proof sketch.}
Substituting the island rule into~\eqref{eq:fourentropy} separates each entropy
into a geometric term and a bulk entropy.  Compatibility identifies the four
bulk domains with one common island $I_\star$, while inclusion-exclusion
cancels the geometric contribution up to
$\eta_{\mathrm{QES}}^{(3)}$, which gives~\eqref{eq:semired}.  The diary belongs
to $I_\star$, so the computable form of Theorem~\ref{thm:bridge} and the
recovery estimate of Ref.~\cite{FawziRenner} bound the remaining bulk
conditional mutual information by
$\delta_{\mathrm{bulk}}^{(3)}+c_3\sqrt{\epsilon(C)}$.  This yields
\eqref{eq:semibridge}. \hfill$\square$

A two-saddle interpolation provides a useful estimate of the compatibility
defect near an exchange.  Assigning to each entropy the soft-min correction
$\log_2(1+2^{-|\Delta_X|})$ gives
\begin{equation}
\begin{aligned}
  \eta_{\mathrm{QES}}^{(3)}(u)
  &\leq
  \sum_{X\in\mathcal X_3}
  \log_2\!\left(1+2^{-|\Delta_X(u)|}\right)\\
  &\leq
  4\log_2\!\left(1+2^{-\Delta_{\min}(u)}\right),
  \qquad
  \Delta_{\min}(u)=\min_{X\in\mathcal X_3}|\Delta_X(u)|.
\end{aligned}
  \label{eq:etaqes}
\end{equation}
Thus the defect is exponentially small in the common large-gap regime and can
be order one during the mixed-saddle crossover.  The interpolation is a
controlled diagnostic of finite-gap sensitivity rather than a substitute for
the joint replica state.

This saddle competition affects the assisted and late-only witnesses differently.
The regions entering $W_R^{(3)}$ all contain the early radiation and therefore
compete between exterior and diary-containing island saddles.  In the
bookkeeping regime of Figure~\ref{fig:jtwindow}, the bins-only regions entering
\begin{equation}
  W_\varnothing^{(3)}
  =I(b_1:b_3\mid b_2)
  \label{eq:unconditioned-witness-semiclassical}
\end{equation}
remain on the exterior branch.  The unconditioned witness is consequently
constant while the conditioned witness changes as the four $R$-assisted
regions exchange saddles.  This is the gravitational origin of the
flat-$\mathcal N_\varnothing$, decreasing-$\mathcal N_R$ structure of the
finite-dimensional model.

For the parity diary,
\begin{equation}
  b_1=X,
  \qquad
  b_2=Y,
  \qquad
  b_3=X\oplus Y,
  \qquad
  X,Y\ \text{independent and uniform},
  \label{eq:semiclassical-parity}
\end{equation}
and therefore
\begin{equation}
  W_\varnothing^{(3)}
  =I(X:X\oplus Y\mid Y)
  =1\ \mathrm{bit}.
  \label{eq:parityplateau}
\end{equation}
Independent thermal factors contribute additively and carry no order-three
conditional mutual information, so the diary record fixes the one-bit
plateau.

For the assisted witness, the staggered saddle exchanges delimit a finite
transition interval. Define
\begin{equation}
  u_{\mathrm{open}}
  =\min_{X\in\mathcal X_3}u_X^\ast,
  \qquad
  u_{\mathrm{close}}
  =\max_{X\in\mathcal X_3}u_X^\ast,
  \label{eq:mixed-saddle-times}
\end{equation}
and, for a selected tolerance $\varepsilon$,
\begin{equation}
  u_{\mathrm{proc}}^{(3)}(\varepsilon)
  =\inf\Bigl\{
      u>u_{\mathrm{close}}:
      W_{R,\mathrm{soft}}^{(3)}(u)\leq\varepsilon
    \Bigr\}.
  \label{eq:threetimes}
\end{equation}
The interval $[u_{\mathrm{open}},u_{\mathrm{close}}]$ is the mixed-saddle
interval, while
$[u_{\mathrm{open}},u_{\mathrm{proc}}^{(3)}(\varepsilon)]$ is the
process-transition interval resolved at tolerance $\varepsilon$.

Two reference times distinguish the entropy and recovery criteria:
\begin{equation}
  S_{\mathrm{rad}}(u_{\mathrm{Page}})
  =S_{\mathrm{BH}}(u_{\mathrm{Page}}),
  \qquad
  S_{\mathrm{rad}}(u_{\mathrm{diary}})
  =S_{\mathrm{BH}}(u_{\mathrm{diary}})+H_M.
  \label{eq:pagerecoverytimes}
\end{equation}
The first is the background entropy Page time.  The second is the
 diary-shifted threshold at which the early radiation carries the
 generalized-entropy budget required to reconstruct the diary.  For the
parameters of Figure~\ref{fig:jtwindow},
\begin{equation}
  u_{\mathrm{open}}
  <u_{\mathrm{Page}}
  <u_{\mathrm{close}}
  <u_{\mathrm{diary}}
  <u_{\mathrm{proc}}^{(3)}(10^{-2}).
  \label{eq:timeordering}
\end{equation}
The entropy Page time therefore marks an interior point of the multi-time
reorganization rather than its unique completion time.  Enlarging the code
subspace shifts the individual exchanges through the same additive entropy
budget that controls the finite-dimensional decoder.  Locally,
\begin{equation}
  \frac{\partial u_X^\ast}
       {\partial\log_2\dim C}
  =\frac{1}{\partial_u\Delta_X(u_X^\ast)},
  \label{eq:codeshift}
\end{equation}
whenever the code-subspace term enters $\Delta_X$ additively and the crossing
is transverse.

\begin{figure}[t]
  \centering
  \includegraphics[width=\textwidth]{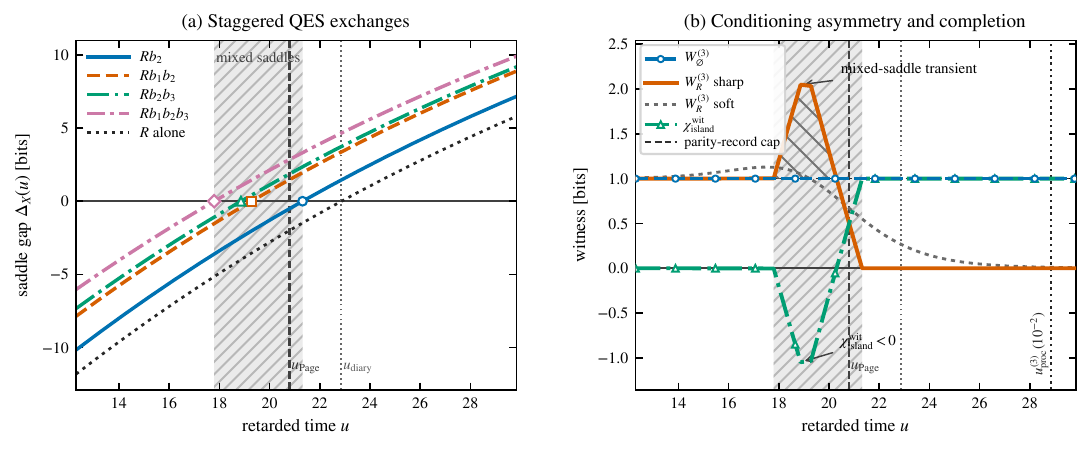}
  \caption{Semiclassical conditioning asymmetry in the entropy-conserving
  exponential bookkeeping model, with
  $S_{\mathrm{BH}}(u)=S_0e^{-u/\tau}$, $S_0=30$ bits, $\tau=30$,
  $\Delta u=1$, and a two-bit parity diary.\ (a) Generalized-entropy gaps of
  the four $R$-assisted regions entering $W_R^{(3)}$.  Their exchange times
  are $17.79$, $18.86$, $19.28$, and $21.32$; the shaded band is the
  mixed-saddle interval.  The background Page time is
  $u_{\mathrm{Page}}=20.79$, while the diary threshold is
  $u_{\mathrm{diary}}=22.86$.  The bins-only regions retain the exterior
  saddle and give the one-bit unconditioned plateau.\ (b) The separately
  minimized sharp-saddle expression reaches $2.05$ bits during the
  mixed-saddle interval.  The horizontal line is the one-bit record cap in
  \eqref{eq:recordcap}.  The two-saddle soft-min surrogate reduces the peak
  to $1.13$ bits and reaches $10^{-2}$ at
  $u_{\mathrm{proc}}^{(3)}=28.85$.  The sharp and soft curves diagnose the
  sensitivity of the crossover to the joint replica treatment; the
  compatible-island regime is governed by Proposition~\ref{prop:semiclassical}.}
  \label{fig:jtwindow}
\end{figure}

\paragraph{The record cap and the mixed-saddle crossover.}
For an independent-emission parity record, write
\begin{equation}
  b_i=C_i\otimes T_i,
  \qquad
  \rho_{Rb_1b_2b_3}
  =\rho_{RC_1C_2C_3}
   \otimes\tau_{T_1}\otimes\tau_{T_2}\otimes\tau_{T_3},
  \label{eq:independentemission}
\end{equation}
where $C_1=X$, $C_2=Y$, $C_3=X\oplus Y$, and the thermal factors are
mutually independent.  Every compatible side register satisfies
\begin{equation}
  W_R^{(3)}
  =I(C_1:C_3\mid C_2R)
  \leq H(C_1\mid C_2R)
  \leq1
  =W_\varnothing^{(3)}.
  \label{eq:recordcap}
\end{equation}
Hence $\chi_{\mathrm{island}}^{\mathrm{wit}}\geq0$ throughout this
factorized record class.

The sharp-saddle overshoot in Figure~\ref{fig:jtwindow} arises because separate
minimization of the four entropies combines dominant saddles
that need not arise from one common replica state during the mixed-saddle
interval.  A persistent excursion above the cap in the joint calculation
would establish fine-grained gravitational correlations across the emitted
bins.  Its suppression would instead assign the excursion to the
compatibility defect $\eta_{\mathrm{QES}}^{(3)}$.  The soft-min curve provides
an illustrative finite-gap bracket and reduces the cap excess from about
$1.05$ bits to about $0.13$ bits. The compatible-island result does not
determine the detailed crossover, which requires the joint replica
calculation.

The saddle and recovery criteria thus distinguish four characteristic times
in the Page transition. At $u_{\mathrm{open}}$ the first region entering the witness
selects a diary-containing island.  At $u_{\mathrm{close}}$ all four selected
saddles are compatible and the witness reduces to the bulk conditional mutual
information in~\eqref{eq:semired}.  At $u_{\mathrm{diary}}$ the early
radiation has the generalized-entropy budget required to reconstruct the
diary on $C$, producing a channel $\mathcal R_{R\to M}$ with error
$\epsilon(C)$.  At $u_{\mathrm{proc}}^{(3)}(\varepsilon)$ the residual
conditioned witness lies below the chosen operational tolerance.

The geometry consequently realizes the screening implication in its physical
form,
\begin{equation}
  \begin{aligned}
    &\bigl(M\ \text{screens the late process}\bigr)
      \quad\text{and}\quad
      \bigl(R\ \text{reconstructs the predictive memory}\bigr)
      \\
    &\hspace{4.5cm}\Longrightarrow
      \bigl(R\ \text{screens the late process}\bigr).
  \end{aligned}
  \label{eq:semiclassical-screening}
\end{equation}
The entropy transition determines when the relevant island saddles become
available, while the multi-time observable determines when their predictive
content is sufficient at the selected temporal order and tolerance.  In this
sense, the semiclassical Page transition is resolved as an interval of
process Markovianization whose completion is controlled by reconstruction
error, saddle compatibility, and the order of the radiation experiment.

\section{Discussion and outlook}
\label{sec:discussion}

The Page transition gives the early radiation a new operational role: it becomes
predictively sufficient for the future multi-time statistics of the Hawking
process. The preceding sections established this statement through the rise of
$\chi_{\mathrm{island}}$, the collapse of the conditional memory cost, and the
forward and converse relations between reconstruction and Markovianization. The
same predictive sufficiency admits three complementary interpretations. The
process wedge organizes which bulk influences the early radiation can simulate
at a chosen temporal resolution, a semiclassical memory budget relates the
required auxiliary memory to QES data, and an algebraic formulation identifies
the form appropriate to gravitational subregion algebras. These structures also
isolate a set of concrete calculations that extend the present results.

\subsection{Predictive reconstruction and the process wedge}
\label{sec:processwedge}

Entanglement-wedge reconstruction concerns the recovery of bulk states or
operators. The process problem asks for the predictive content required to
reproduce the response of the future radiation to interventions. The two tasks
coincide when the microscopic state is fixed by its future influence, while
finite temporal resolution naturally organizes predictive reconstruction into
a filtration.

Let $X$ be a bulk memory register in a finite-dimensional truncation, let
$\mathcal J$ range over the admissible interventions on $X$, and write
$\Upsilon_L[\mathcal J_X;\rho]$ for the induced late comb when the code state is
$\rho\in C$. Two registers are predictively equivalent at resolution $K$ when
no tester spanning at most $K$ late steps distinguishes their influence,
\begin{equation}
  X\sim_K X'
  \quad\Longleftrightarrow\quad
  \sup_{\mathcal J}\sup_{\rho\in C}
  d_{\mathrm{op}}^{(K)}\!\left(
    \Upsilon_L[\mathcal J_X;\rho],
    \Upsilon_L[\mathcal J_{X'};\rho]
  \right)=0.
  \label{eq:pw-equivalence}
\end{equation}
The class $[X]_K$ is the predictive content of $X$ visible at that resolution.

\begin{definition}[Process wedge at resolution $K$]
\label{def:pw}
Fix an instrument family $\mathcal I$, a code subspace $C$, a tolerance
$\varepsilon$, and a temporal resolution $K$. The class $[X]_K$ belongs to
$\mathrm{PW}^{(K),\varepsilon}(R;C,\mathcal I)$ when there exists a single
channel $\mathcal S_{R\to\widehat X}$ such that
\begin{equation}
  \sup_{\mathcal J}\sup_{\rho\in C}
  d_{\mathrm{op}}^{(K)}\!\left(
    \Upsilon_L[\mathcal J_X;\rho],
    \Upsilon_L[\mathcal J_{\widehat X};\rho]
  \right)
  \leq \varepsilon,
  \qquad
  \widehat X=\mathcal S_{R\to\widehat X}(R).
  \label{eq:pwdef}
\end{equation}
Thus the process wedge contains the bulk influence classes whose effects on the
late multi-time experiment are simulable from the early radiation.
\end{definition}

The dependence on $K$ is part of the construction. Agreement at resolution
$K'$ implies agreement for every shorter tester, so the natural projection from
$K'$-resolved classes to $K$-resolved classes maps
\begin{equation}
  \mathrm{PW}^{(K'),\varepsilon}(R;C,\mathcal I)
  \longrightarrow
  \mathrm{PW}^{(K),\varepsilon}(R;C,\mathcal I),
  \qquad K\leq K'.
  \label{eq:pw-filtration}
\end{equation}
The full-resolution wedge $\mathrm{PW}^{\varepsilon}(R;C,\mathcal I)$ consists
of the influence simulable by one channel at every temporal order. In the
single-order parity family, the memory class agrees with the trivial class for
$K<k_\ast$ and becomes distinguishable at $K=k_\ast$; after recovery from $R$,
the same class is simulable at every $K$. The saturation order therefore has a
geometric reading as the first temporal resolution at which the unsupplied
predictive memory becomes operationally visible.

Reconstruction supplies a robust relation to the entanglement wedge. Let $q_K$
map a reconstructible bulk register to its $K$-resolved influence class. If
$X$ is recoverable from $R$ on $C$ with entanglement infidelity
$\epsilon_X(C)$, and this recovery controls the induced process uniformly over
the admissible interventions, the recovery channel followed by $\mathcal J$
simulates the influence of $X$. Data processing and the Fuchs--van de Graaf
bound then give
\begin{equation}
  q_K\!\left(EW^{\epsilon_X}(R;C)\right)
  \subseteq
  \mathrm{PW}^{(K),\varepsilon_X}(R;C,\mathcal I),
  \qquad
  \varepsilon_X
  =c_{\mathcal I,K}\sqrt{\epsilon_X(C)},
  \label{eq:ew-pw-inclusion}
\end{equation}
with $c_{\mathcal I,K}$ fixed by the tester normalization and the number of
late slots. This inclusion is the influence-level form of the forward bridge.
After the Page transition, entanglement-wedge reconstruction places the
predictive island memory in every finite-resolution process wedge.

The parity construction also clarifies the object reconstructed by the
converse theorem. Its static common-cause description may contain a larger
record than the minimal memory carried by the causal dilation, while the future
late statistics depend only on $M_{\min}$. The process wedge therefore quotients
microscopic registers by predictive equivalence and asks for a simulator of the
resulting class. This is the operational content of the hierarchy
$M_{\min}\subseteq M$: simulation of the future influence requires the minimal
predictive memory rather than recovery of every microscopic degree of freedom.
The term ``wedge'' is used by analogy, since the object is a quotient of bulk
degrees of freedom by experimentally accessible influence rather than a
geometric subregion. Its relation to one-shot min and max entanglement wedges
at matched error and instrument constraints provides a natural comparison with
one-shot holography~\cite{AkersPenington,OneShotHolography}.

\subsection{A semiclassical memory budget}
\label{sec:memorybudget}

The conditional memory cost gives a resource statement of the transition. It
counts the smallest predictive register that must accompany $R$ to reduce the
late memory below the tolerance $\epsilon$. A one-shot description of a
minimal dilation suggests
\begin{equation}
  m_\epsilon(L\,|\,R)
  \lesssim
  \max\!\left\{
    H_{\max}^{\epsilon'}(M_{\min}|R),0
  \right\}
  +O\!\left(\log\frac{1}{\epsilon-\epsilon'}\right),
  \label{eq:memory-one-shot}
\end{equation}
so the gravitational question is how the QES prescription estimates the
conditional max entropy of the predictive memory~\cite{BisioMemory,Berta}.

For semiclassical states admitting a compressible finite-dimensional
predictive sector, define the positive QES memory budget
\begin{equation}
  \mathcal B_{\mathrm{QES}}(t)
  =\left[
    \frac{\Delta A_{\mathrm{QES}}(t)}{4G_N}
    +S_{\mathrm{bulk}}^{\mathrm{mem}}(t)
  \right]_+.
  \label{eq:qes-budget}
\end{equation}
Here $\Delta A_{\mathrm{QES}}(t)$ is the area contribution separating the
radiation-accessible saddle from the predictive-island saddle relevant to the
memory, and $S_{\mathrm{bulk}}^{\mathrm{mem}}(t)$ is the entropy of the bulk
degrees of freedom carrying that memory across the cut. Gravitational
entropies in~\eqref{eq:qes-budget} are measured in natural logarithms, while
$m_\epsilon$ counts qubits.

\emph{Conjecture (heuristic).} In this regime,
\begin{equation}
  m_\epsilon(L\,|\,R_t)
  \ \lesssim\ 
  \frac{\mathcal B_{\mathrm{QES}}(t)}{\ln 2}
  +\Delta_{\mathrm{1shot}}(\epsilon,\psi),
  \label{eq:memory-budget-conjecture}
\end{equation}
where $\Delta_{\mathrm{1shot}}$ is a smoothing and finite-block correction
that vanishes in the asymptotic limit. The finite-dimensional evaporation comb
fixes the bulk normalization: before recovery its cost is
$\log_2\dim M_{\min}$, and after recovery it is zero at the selected tolerance.
The JT calculation supplies the QES exchange and the
$\sqrt{\epsilon(C)}$ reconstruction correction. The additional step in
\eqref{eq:memory-budget-conjecture}, identifying the one-shot predictive
entropy with the generalized-entropy budget, is the holographic input. Its
proof would turn the collapse of $m_\epsilon$ into a resource-theoretic
refinement of the Page curve, with the dilaton and bulk-memory terms providing
the expected semiclassical scaling~\cite{JafferisKolchmeyer}.

\subsection{Algebraic formulation}
\label{sec:algebraic}

The island and radiation are naturally described by operator algebras. Let
$\mathcal A_I$ denote the predictive island algebra and $\mathcal A_R$ the
algebra available in the early radiation. A conditional expectation
$E_R:\mathcal A\to\mathcal A_R$ implements restriction and coarse-graining,
while a recovery morphism
$\beta^\dagger:\mathcal A_I\to\mathcal A_R$ represents island observables in
the radiation algebra. These maps play distinct roles: $E_R$ specifies the
accessible data, and the Petz-type recovery associated with $\beta^\dagger$
reconstructs the predictive algebra from those data.

The algebraic forward statement is approximate sufficiency. If
$\mathcal A_I$ screens the future late algebra from the earlier history and
$\mathcal A_R$ reconstructs $\mathcal A_I$ uniformly on the code subspace,
then the predictive functional of the late process factors through
$\mathcal A_R$ up to the reconstruction and residual-memory errors. In a
finite tensor product, a state-preserving conditional expectation acts as the
identity on $\mathcal A_R$ and averages the complementary algebra against a
reference state; its predual is the corresponding coarse-graining channel.
The finite-dimensional bridge theorem is this tensor-factor instance, with
the recovery map kept separate from the expectation.

Petz sufficiency is formulated directly for subalgebras, and crossed-product
constructions provide type-$\mathrm{II}$ algebras with traces and generalized
entropies in gravitational subregions~\cite{Petz,LeutheusserLiu,WittenCrossed,CPW}.
The present proofs use finite-dimensional factors. Extending them to
type-$\mathrm{II}$ and type-$\mathrm{III}$ algebras requires a tester distance
for energy-bounded local operations and a continuity estimate controlling the
induced multi-time functional in a completely bounded or operational norm.

\subsection{Open calculations and extensions}
\label{sec:outlook}

\paragraph{The joint replica state.}
During the mixed-saddle interval, separate minimization of the four entropies
can combine saddles that do not arise from one joint replica state. The
parity-record cap therefore gives a sharp diagnostic. Persistence of the
excursion above one bit would reveal fine-grained gravitational correlations
among the time bins, whereas suppression would identify it with the
compatibility defect $\eta_{\mathrm{QES}}^{(3)}$. A joint replica calculation
for the time-binned diary background, naturally related to end-of-world-brane
models of replica wormholes~\cite{PSSY}, should evaluate all four regions in
one replicated geometry. Its output is the physical crossover curve, the
ordering of the compatible QES exchanges, and the tolerance-resolved
completion time of the process transition.

\paragraph{The one-shot QES bridge.}
Equation~\eqref{eq:memory-budget-conjecture} reduces the geometric problem to a
specific information-theoretic comparison: estimate
$H_{\max}^{\epsilon'}(M_{\min}|R)$ from the generalized-entropy data of the
competing saddles. A proof requires a one-shot version of predictive
reconstruction in which the code-subspace error, the smoothing parameter, and
the instrument-dependent continuity modulus appear in the same bound. Such a
result would determine whether the area term controls the entire predictive
memory or only the part visible to a selected family of late-time operations.
It would also place the process wedge relative to the min and max entanglement
wedges at matched tolerances.

\paragraph{Complexity-filtered prediction.}
Definition~\ref{def:pw} allows an unrestricted simulator. Restricting the
circuit complexity of $\mathcal S_{R\to\widehat X}$ defines a family
$\mathrm{PW}_{\mathcal C}^{(K),\varepsilon}$, while restricting the tester
complexity defines a separate notion of experimentally resolvable influence.
Black-hole decoding and python's-lunch geometry suggest that the island's
predictive influence may become simulable only above a parametrically large
complexity~\cite{HarlowHayden,PythonsLunch}. The information-theoretic Page
transition and the computationally accessible transition can therefore occur
at distinct operational thresholds. Determining the simulator threshold, the
tester threshold, and their relation to nonminimal QESs would add a
computational axis to the temporal-resolution filtration and define a
complexity-restricted version of $\chi_{\mathrm{island}}$.

\paragraph{Geometric, multipartite, and algebraic extensions.}
The covariant construction carries over whenever the radiation bins are
associated with causal diamonds along an appropriate collection worldline.
Rotating and Kerr evaporation test how frame dragging and superradiant sectors
modify the predictive register; wedge-holographic realizations provide a brane
interpretation of the process wedge; and distributed radiation registers turn
the conditioning problem into a multipartite simulation and secret-sharing
task~\cite{KerrIsland,RotatingBTZIsland,WedgeHolography,SecretSharing}. The
algebraic extension replaces finite registers by energy-bounded local
operations on subregion algebras. In each setting the geometry and admissible
operations change, while $\chi_{\mathrm{island}}$ continues to measure the
temporal memory removed by access to the corresponding recovery algebra.

Across the Page transition the late radiation continues to carry temporal
memory, while the early radiation becomes sufficient to screen its predictive
influence. The entropy curve identifies the reorganization of the encoding;
the process observable resolves the temporal order at which that encoding
controls future experiments. This is the operational content of the Page
transition as a Markovianization transition.

\section*{Data and code availability}
The data and code supporting this work are publicly available on Zenodo under
DOI~\url{https://doi.org/10.5281/zenodo.22813260}.

\bibliographystyle{JHEP}
\bibliography{biblio.bib}

@article{priorpaper,
  author        = {Tushar Waghmare},
  title         = {Covariant Measures of Non-{M}arkovianity in Curved Spacetime},
  journal       = {Class. Quantum Grav.},
  doi           = {10.1088/1361-6382/ae7d7c},
  volume        = {43},
  number        = {13},
  year          = {2026},
  pages         = {135003},
  eprint        = {2511.15365},
  archivePrefix = {arXiv},
  primaryClass  = {hep-th},
  url           = {https://doi.org/10.1088/1361-6382/ae7d7c},
}

@article{Pollock,
  author        = {Pollock, Felix A. and Rodr{\'\i}guez-Rosario, C{\'e}sar and
                   Frauenheim, Thomas and Paternostro, Mauro and Modi, Kavan},
  title         = {Non-{M}arkovian quantum processes: complete framework and
                   efficient characterization},
  journal       = {Phys. Rev. A},
  volume        = {97},
  pages         = {012127},
  year          = {2018},
  eprint        = {1512.00589},
  archivePrefix = {arXiv},
  primaryClass  = {quant-ph},
  doi           = {10.1103/PhysRevA.97.012127},
  url           = {https://doi.org/10.1103/PhysRevA.97.012127},
}

@article{Taranto,
  author        = {Taranto, Philip and Pollock, Felix A. and Milz, Simon and
                   Tomamichel, Marco and Modi, Kavan},
  title         = {Quantum {M}arkov Order},
  journal       = {Phys. Rev. Lett.},
  volume        = {122},
  pages         = {140401},
  year          = {2019},
  eprint        = {1805.11341},
  archivePrefix = {arXiv},
  primaryClass  = {quant-ph},
  doi           = {10.1103/PhysRevLett.122.140401},
  url           = {https://doi.org/10.1103/PhysRevLett.122.140401},
}

@article{CDP,
  author        = {Chiribella, Giulio and D'Ariano, Giacomo Mauro and
                   Perinotti, Paolo},
  title         = {Theoretical framework for quantum networks},
  journal       = {Phys. Rev. A},
  volume        = {80},
  pages         = {022339},
  year          = {2009},
  eprint        = {0904.4483},
  archivePrefix = {arXiv},
  primaryClass  = {quant-ph},
  doi           = {10.1103/PhysRevA.80.022339},
  url           = {https://doi.org/10.1103/PhysRevA.80.022339},
}

@inproceedings{GutoskiWatrous,
  author        = {Gutoski, Gus and Watrous, John},
  title         = {Toward a general theory of quantum games},
  booktitle     = {Proceedings of the 39th Annual ACM Symposium on Theory of
                   Computing (STOC)},
  pages         = {565--574},
  year          = {2007},
  eprint        = {quant-ph/0611234},
  archivePrefix = {arXiv},
  doi           = {10.1145/1250790.1250873},
  url           = {https://doi.org/10.1145/1250790.1250873},
}

@article{FawziRenner,
  author        = {Fawzi, Omar and Renner, Renato},
  title         = {Quantum conditional mutual information and approximate
                   {M}arkov chains},
  journal       = {Commun. Math. Phys.},
  volume        = {340},
  number        = {2},
  pages         = {575--611},
  year          = {2015},
  eprint        = {1410.0664},
  archivePrefix = {arXiv},
  primaryClass  = {quant-ph},
  doi           = {10.1007/s00220-015-2466-x},
  url           = {https://doi.org/10.1007/s00220-015-2466-x},
}

@article{JungeUniversal,
  author        = {Junge, Marius and Renner, Renato and Sutter, David and
                   Wilde, Mark M. and Winter, Andreas},
  title         = {Universal recovery maps and approximate sufficiency of
                   quantum relative entropy},
  journal       = {Ann. Henri Poincar\'e},
  volume        = {19},
  number        = {10},
  pages         = {2955--2978},
  year          = {2018},
  eprint        = {1509.07127},
  archivePrefix = {arXiv},
  primaryClass  = {quant-ph},
  doi           = {10.1007/s00023-018-0716-0},
  url           = {https://doi.org/10.1007/s00023-018-0716-0},
}

@article{Petz,
  author        = {Petz, D{\'e}nes},
  title         = {Sufficient subalgebras and the relative entropy of states of
                   a von {N}eumann algebra},
  journal       = {Commun. Math. Phys.},
  volume        = {105},
  number        = {1},
  pages         = {123--131},
  year          = {1986},
  doi           = {10.1007/BF01212345},
  url           = {https://doi.org/10.1007/BF01212345},
}

@article{HaydenPreskill,
  author        = {Hayden, Patrick and Preskill, John},
  title         = {Black holes as mirrors: quantum information in random
                   subsystems},
  journal       = {JHEP},
  volume        = {09},
  pages         = {120},
  year          = {2007},
  eprint        = {0708.4025},
  archivePrefix = {arXiv},
  primaryClass  = {hep-th},
  doi           = {10.1088/1126-6708/2007/09/120},
  url           = {https://doi.org/10.1088/1126-6708/2007/09/120},
}

@article{Penington,
  author        = {Penington, Geoffrey},
  title         = {Entanglement wedge reconstruction and the information paradox},
  journal       = {JHEP},
  volume        = {09},
  pages         = {002},
  year          = {2020},
  eprint        = {1905.08255},
  archivePrefix = {arXiv},
  primaryClass  = {hep-th},
  doi           = {10.1007/JHEP09(2020)002},
  url           = {https://doi.org/10.1007/JHEP09(2020)002},
}

@article{AEMM,
  author        = {Almheiri, Ahmed and Engelhardt, Netta and Marolf, Donald and
                   Maxfield, Henry},
  title         = {The entropy of bulk quantum fields and the entanglement wedge
                   of an evaporating black hole},
  journal       = {JHEP},
  volume        = {12},
  pages         = {063},
  year          = {2019},
  eprint        = {1905.08762},
  archivePrefix = {arXiv},
  primaryClass  = {hep-th},
  doi           = {10.1007/JHEP12(2019)063},
  url           = {https://doi.org/10.1007/JHEP12(2019)063},
}

@article{LeutheusserLiu,
  author        = {Leutheusser, Samuel and Liu, Hong},
  title         = {Emergent times in holographic duality},
  journal       = {Phys. Rev. D},
  volume        = {108},
  number        = {8},
  year          = {2023},
  pages         = {086020},
  eprint        = {2112.12156},
  archivePrefix = {arXiv},
  primaryClass  = {hep-th},
  doi           = {10.1103/PhysRevD.108.086020},
  url           = {https://doi.org/10.1103/PhysRevD.108.086020},
}

@article{WittenCrossed,
  author        = {Witten, Edward},
  title         = {Gravity and the crossed product},
  journal       = {JHEP},
  volume        = {10},
  year          = {2022},
  pages         = {008},
  eprint        = {2112.12828},
  archivePrefix = {arXiv},
  primaryClass  = {hep-th},
  doi           = {10.1007/JHEP10(2022)008},
  url           = {https://doi.org/10.1007/JHEP10(2022)008},
}

@article{CPW,
  author        = {Chandrasekaran, Venkatesa and Penington, Geoff and Witten, Edward},
  title         = {Large {N} algebras and generalized entropy},
  journal       = {JHEP},
  volume        = {04},
  year          = {2023},
  pages         = {009},
  eprint        = {2209.10454},
  archivePrefix = {arXiv},
  primaryClass  = {hep-th},
  doi           = {10.1007/JHEP04(2023)009},
  url           = {https://doi.org/10.1007/JHEP04(2023)009},
}

@article{AkersPenington,
  author        = {Akers, Chris and Penington, Geoff},
  title         = {Leading order corrections to the quantum extremal surface
                   prescription},
  journal       = {JHEP},
  volume        = {04},
  year          = {2021},
  pages         = {062},
  eprint        = {2008.03319},
  archivePrefix = {arXiv},
  primaryClass  = {hep-th},
  doi           = {10.1007/JHEP04(2021)062},
  url           = {https://doi.org/10.1007/JHEP04(2021)062},
}

@article{OneShotHolography,
  author        = {Akers, Chris and Levine, Adam and Penington, Geoff and
                   Wildenhain, Elizabeth},
  title         = {One-shot holography},
  journal       = {SciPost Phys.},
  volume        = {16},
  year          = {2024},
  pages         = {144},
  eprint        = {2307.13032},
  archivePrefix = {arXiv},
  primaryClass  = {hep-th},
  doi           = {10.21468/SciPostPhys.16.6.144},
  url           = {https://doi.org/10.21468/SciPostPhys.16.6.144},
}

@mastersthesis{Berta,
  author        = {Berta, Mario},
  title         = {Single-shot quantum state merging},
  type          = {Diploma thesis},
  school        = {ETH Zurich},
  year          = {2008},
  eprint        = {0912.4495},
  archivePrefix = {arXiv},
  primaryClass  = {quant-ph},
  url           = {https://arxiv.org/abs/0912.4495},
}

@article{BisioMemory,
  author        = {Bisio, Alessandro and D'Ariano, Giacomo Mauro and
                   Perinotti, Paolo and Sedl\'ak, Michal},
  title         = {Memory cost of quantum protocols},
  journal       = {Phys. Rev. A},
  volume        = {85},
  number        = {3},
  year          = {2012},
  pages         = {032333},
  eprint        = {1112.3853},
  archivePrefix = {arXiv},
  primaryClass  = {quant-ph},
  doi           = {10.1103/PhysRevA.85.032333},
  url           = {https://doi.org/10.1103/PhysRevA.85.032333},
}

@article{JafferisKolchmeyer,
  author        = {Jafferis, Daniel Louis and Kolchmeyer, David K.},
  title         = {Entanglement entropy in Jackiw-Teitelboim gravity},
  year          = {2019},
  eprint        = {1911.10663},
  archivePrefix = {arXiv},
  primaryClass  = {hep-th},
  url           = {https://arxiv.org/abs/1911.10663},
}

@article{SecretSharing,
  author        = {Balasubramanian, Vijay and Kar, Arjun and Parrikar, Onkar and
                   S\'arosi, G\'abor and Ugajin, Tomonori},
  title         = {Geometric secret sharing in a model of Hawking radiation},
  journal       = {JHEP},
  volume        = {01},
  year          = {2021},
  pages         = {177},
  eprint        = {2003.05448},
  archivePrefix = {arXiv},
  primaryClass  = {hep-th},
  doi           = {10.1007/JHEP01(2021)177},
  url           = {https://doi.org/10.1007/JHEP01(2021)177},
}

@article{WedgeHolography,
  author        = {Akal, Ibrahim and Kusuki, Yuya and Takayanagi, Tadashi and
                   Wei, Zixia},
  title         = {Codimension-two holography for wedges},
  journal       = {Phys. Rev. D},
  volume        = {102},
  number        = {12},
  year          = {2020},
  pages         = {126007},
  eprint        = {2007.06800},
  archivePrefix = {arXiv},
  primaryClass  = {hep-th},
  doi           = {10.1103/PhysRevD.102.126007},
  url           = {https://doi.org/10.1103/PhysRevD.102.126007},
}

@article{KerrIsland,
  author        = {Wang, Liqiang and Li, Ran},
  title         = {Entanglement islands and the Page curve of Hawking radiation for
                   rotating Kerr black holes},
  journal       = {Phys. Rev. D},
  volume        = {110},
  number        = {6},
  year          = {2024},
  pages         = {066012},
  eprint        = {2406.13949},
  archivePrefix = {arXiv},
  primaryClass  = {hep-th},
  doi           = {10.1103/PhysRevD.110.066012},
  url           = {https://doi.org/10.1103/PhysRevD.110.066012},
}

@article{RotatingBTZIsland,
  author        = {Yu, Ming-Hui and Lu, Cheng-Yuan and Ge, Xian-Hui and Sin, Sang-Jin},
  title         = {Island, Page curve, and superradiance of rotating {BTZ} black holes},
  journal       = {Phys. Rev. D},
  volume        = {105},
  number        = {6},
  year          = {2022},
  pages         = {066009},
  eprint        = {2112.14361},
  archivePrefix = {arXiv},
  primaryClass  = {hep-th},
  doi           = {10.1103/PhysRevD.105.066009},
  url           = {https://doi.org/10.1103/PhysRevD.105.066009},
}

@article{Hollowood:2020cou,
  author        = {Hollowood, Timothy J. and Kumar, S. Prem and Legramandi, Andrea},
  title         = {Hawking radiation correlations of evaporating black holes in JT gravity},
  journal       = {J. Phys. A},
  volume        = {53},
  number        = {47},
  pages         = {475401},
  year          = {2020},
  doi           = {10.1088/1751-8121/abbc51},
  eprint        = {2007.04877},
  archivePrefix = {arXiv},
  primaryClass  = {hep-th},
  url           = {https://doi.org/10.1088/1751-8121/abbc51},
}

@article{Hollowood:2020couPage,
  author        = {Hollowood, Timothy J. and Kumar, S. Prem},
  title         = {Islands and Page curves for evaporating black holes in JT gravity},
  journal       = {JHEP},
  volume        = {08},
  pages         = {094},
  year          = {2020},
  doi           = {10.1007/JHEP08(2020)094},
  eprint        = {2004.14944},
  archivePrefix = {arXiv},
  primaryClass  = {hep-th},
  url           = {https://doi.org/10.1007/JHEP08(2020)094},
}

@article{PSSY,
  author        = {Penington, Geoff and Shenker, Stephen H. and Stanford, Douglas and
                   Yang, Zhenbin},
  title         = {Replica wormholes and the black hole interior},
  journal       = {JHEP},
  volume        = {03},
  pages         = {205},
  year          = {2022},
  eprint        = {1911.11977},
  archivePrefix = {arXiv},
  primaryClass  = {hep-th},
  doi           = {10.1007/JHEP03(2022)205},
  url           = {https://doi.org/10.1007/JHEP03(2022)205},
}

@article{HarlowHayden,
  author        = {Harlow, Daniel and Hayden, Patrick},
  title         = {Quantum computation vs. firewalls},
  journal       = {JHEP},
  volume        = {06},
  pages         = {085},
  year          = {2013},
  eprint        = {1301.4504},
  archivePrefix = {arXiv},
  primaryClass  = {hep-th},
  doi           = {10.1007/JHEP06(2013)085},
  url           = {https://doi.org/10.1007/JHEP06(2013)085},
}

@article{PythonsLunch,
  author        = {Brown, Adam R. and Gharibyan, Hrant and Penington, Geoff and
                   Susskind, Leonard},
  title         = {The {P}ython's {L}unch: geometric obstructions to decoding
                   {H}awking radiation},
  journal       = {JHEP},
  volume        = {08},
  pages         = {121},
  year          = {2020},
  eprint        = {1912.00228},
  archivePrefix = {arXiv},
  primaryClass  = {hep-th},
  doi           = {10.1007/JHEP08(2020)121},
  url           = {https://doi.org/10.1007/JHEP08(2020)121},
}

@article{PageInformation,
  author        = {Page, Don N.},
  title         = {Information in black hole radiation},
  journal       = {Phys. Rev. Lett.},
  volume        = {71},
  pages         = {3743--3746},
  year          = {1993},
  eprint        = {hep-th/9306083},
  archivePrefix = {arXiv},
  doi           = {10.1103/PhysRevLett.71.3743},
  url           = {https://doi.org/10.1103/PhysRevLett.71.3743},
}

@article{EngelhardtWall,
  author        = {Engelhardt, Netta and Wall, Aron C.},
  title         = {Quantum extremal surfaces: holographic entanglement entropy
                   beyond the classical regime},
  journal       = {JHEP},
  volume        = {01},
  pages         = {073},
  year          = {2015},
  eprint        = {1408.3203},
  archivePrefix = {arXiv},
  primaryClass  = {hep-th},
  doi           = {10.1007/JHEP01(2015)073},
  url           = {https://doi.org/10.1007/JHEP01(2015)073},
}

@article{AMMZ,
  author        = {Almheiri, Ahmed and Mahajan, Raghu and Maldacena, Juan and
                   Zhao, Ying},
  title         = {The Page curve of {H}awking radiation from semiclassical geometry},
  journal       = {JHEP},
  volume        = {03},
  pages         = {149},
  year          = {2020},
  eprint        = {1908.10996},
  archivePrefix = {arXiv},
  primaryClass  = {hep-th},
  doi           = {10.1007/JHEP03(2020)149},
  url           = {https://doi.org/10.1007/JHEP03(2020)149},
}

@article{CotlerUniversal,
  author        = {Cotler, Jordan and Hayden, Patrick and Penington, Geoffrey and
                   Salton, Grant and Swingle, Brian and Walter, Michael},
  title         = {Entanglement wedge reconstruction via universal recovery channels},
  journal       = {Phys. Rev. X},
  volume        = {9},
  number        = {3},
  pages         = {031011},
  year          = {2019},
  eprint        = {1704.05839},
  archivePrefix = {arXiv},
  primaryClass  = {hep-th},
  doi           = {10.1103/PhysRevX.9.031011},
  url           = {https://doi.org/10.1103/PhysRevX.9.031011},
}

\appendix
\clearpage
\section*{Appendix}
\addcontentsline{toc}{section}{Appendix}

\appendix

\section{Finite-order constructions}
\label{app:finiteorder}

This appendix gives the explicit constructions underlying the finite-order
statements in the main text.  The parity comb establishes the separation
between low-order marginals and the full multi-time process, its conditioned
extension fixes the endpoints of the island witness, and its causal dilation
identifies the minimal predictive memory that enters the conditional memory
cost.

\subsection{The parity comb and its reduced processes}
\label{app:parity}

Let $X$ and $Y$ be independent uniform bits and define
\begin{equation}
  b_1=X,
  \qquad
  b_2=Y,
  \qquad
  b_3=X\oplus Y.
  \label{eq:app-parity-map}
\end{equation}
A measure-and-prepare realization carries the classical record
$c=(X,Y)$ and uses, at slot $i$, the completely positive map
$\mathcal E_i^{(c)}$ that discards its input and prepares
$|b_i(c)\rangle\langle b_i(c)|$.  The resulting process is
\begin{equation}
  \Upsilon_L
  =\sum_{c\in\{0,1\}^2}p(c)
    \bigotimes_{i=1}^{3}\mathcal E_i^{(c)},
  \qquad
  p(c)=\frac14.
  \label{eq:app-parity-comb}
\end{equation}
The common label makes the full comb non-Markovian even though every
conditional component in the mixture is a product process.

For a subset $S\subseteq\{1,2,3\}$, the reduced comb is obtained by applying
the trace instrument at the complementary slots,
\begin{equation}
  \Upsilon_L\big|_S
  =\sum_c p(c)\bigotimes_{i\in S}\mathcal E_i^{(c)}.
  \label{eq:app-reduced-comb}
\end{equation}
Each single variable is uniform.  Every pair is also independent and uniform:
for example,
\begin{equation}
  p(b_1=x,b_3=z)
  =\sum_y \frac14\,\mathbf 1[z=x\oplus y]
  =\frac14,
  \label{eq:app-pairwise-independent}
\end{equation}
and the same calculation applies to the other pairs.  Hence, for
$|S|\leq2$,
\begin{equation}
  \Upsilon_L\big|_S
  =\Lambda\big|_S,
  \label{eq:app-low-order-equality}
\end{equation}
where $\Lambda$ is the factorized Markov comb that prepares an independent
maximally mixed bit at every slot.  Equality holds as a multilinear
functional, not only as an equality of unmeasured output states.  Each $\mathcal E_i^{(c)}$ discards the incoming system, so an intervention at
one slot leaves the preparation at every other slot invariant; adaptive
choices therefore do not distinguish the two reduced combs.

At order three the mixture does not factorize.  Its outcome distribution is
\begin{equation}
  p(b_1,b_2,b_3)
  =\frac14\,
    \mathbf 1[b_1\oplus b_2\oplus b_3=0],
  \label{eq:app-parity-sector}
\end{equation}
which differs from the uniform product.  Equations \eqref{eq:app-low-order-equality} and
\eqref{eq:app-parity-sector} establish items 1 and 2 of the $k_\ast=3$
construction in Theorem~\ref{thm:blindness};
Appendix~\ref{app:endpoints} completes item 3.

\subsection{Parity memory at arbitrary order}
\label{app:higherparity}

The same separation can be placed at any prescribed order $K_0\geq3$.  Let
$X_1,\ldots,X_{K_0-1}$ be independent uniform bits and take
\begin{equation}
  b_i=X_i\quad(1\leq i<K_0),
  \qquad
  b_{K_0}=\bigoplus_{i=1}^{K_0-1}X_i.
  \label{eq:app-higher-parity}
\end{equation}
The full distribution is uniform on one parity sector.  Every proper marginal
is uniform: if a variable $X_j$ is omitted, summing over it maps each fixed
configuration of the retained bins to exactly one value of the parity bit and
produces the factor $1/2$ required by the uniform distribution.  If the parity
bin itself is omitted, the remaining variables are independent by
construction.  Since the measure-and-prepare maps discard their inputs, the
same statement holds for every reduced process as a multilinear functional.
Thus
\begin{equation}
  \Upsilon_L\big|_S=\Lambda\big|_S
  \quad\text{for every }|S|<K_0,
  \qquad
  \Upsilon_L\neq\Lambda
  \quad\text{at }|S|=K_0.
  \label{eq:app-higher-order-separation}
\end{equation}
A causal implementation stores a running parity bit.  After slot $j$, the
memory is $P_j=X_1\oplus\cdots\oplus X_j$; the next fresh seed updates
$P_j\mapsto P_{j+1}$.  One bit therefore suffices for every $K_0$, while the
first visible and saturation orders are both $K_0$ for this single-order
family.

\subsection{Endpoint values and the conditioned witness}
\label{app:endpoints}

For the three-bin parity distribution,
\begin{equation}
  W_\varnothing^{(3)}
  =I(b_1:b_3\mid b_2).
  \label{eq:app-unconditioned-witness}
\end{equation}
Because $b_1$ and $b_3$ are individually uniform conditional on $b_2$, while
the pair $(b_1,b_3)$ is fixed by one uniform bit,
\begin{equation}
  H(b_1\mid b_2)=1,
  \qquad
  H(b_3\mid b_2)=1,
  \qquad
  H(b_1b_3\mid b_2)=1,
\end{equation}
and hence
\begin{equation}
  W_\varnothing^{(3)}=1\ \mathrm{bit}.
  \label{eq:app-one-bit-witness}
\end{equation}
For passive classical statistics and relative entropy, this conditional mutual
information is exactly the distance to the nearest conditional product
$p(b_1,b_2)p(b_3\mid b_2)$.  For other operational divergences it remains a
computable diagnostic of the same order-three correlation.

Let the early register encode $X$ through two equiprobable pure states
$|r_0\rangle$ and $|r_1\rangle$ with overlap
$s=|\langle r_0|r_1\rangle|$.  The conditional witness is
\begin{equation}
  W_R^{(3)}
  =I(X:X\oplus Y\mid YR)
  =H(X\mid R).
  \label{eq:app-conditioned-parity-reduction}
\end{equation}
For the binary classical-quantum state
$\rho_{XR}=\frac12\sum_x|x\rangle\langle x|\otimes
|r_x\rangle\langle r_x|$, the average state on $R$ has eigenvalues
$(1\pm s)/2$.  Therefore
\begin{equation}
  H(X\mid R)
  =1-h_2\!\left(\frac{1+s}{2}\right),
  \qquad
  \chi_{\mathrm{island}}^{\mathrm{wit},(3)}
  =h_2\!\left(\frac{1+s}{2}\right).
  \label{eq:app-conditioned-parity-closed}
\end{equation}
At the no-recovery endpoint $s=1$, the early register is independent of the
predictive bit and $W_R^{(3)}=W_\varnothing^{(3)}=1$.  At exact recovery,
$s=0$, the states are orthogonal and $W_R^{(3)}=0$.  The late-only witness is
independent of the recovery coordinate, so the two endpoints give
\begin{equation}
  \chi_{\mathrm{island}}^{\mathrm{wit},(3)}=0
  \quad\text{before recovery},
  \qquad
  \chi_{\mathrm{island}}^{\mathrm{wit},(3)}=1\ \mathrm{bit}
  \quad\text{after exact recovery}.
  \label{eq:app-witness-endpoints}
\end{equation}
The two-bin screening core follows from the same calculation with
$I(b_1:b_2\mid R)$ in place of \eqref{eq:app-conditioned-parity-reduction}.

\subsection{The minimal screening register}
\label{app:screening}

The static record $(X,Y)$ contains two bits, while the causal dilation needs
only one bit of predictive memory.  After emitting $b_1=X$, the dilation
retains $X$.  At the second slot it draws the fresh seed $Y=b_2$ and updates
the carried memory to $X\oplus Y$, which is emitted as $b_3$.  Conditioning on
$X$ screens the distant past,
\begin{equation}
  I(b_1:b_3\mid b_2,X)=0.
  \label{eq:app-X-screens}
\end{equation}
By contrast, supplying the middle value $Y$ as an additional register adds no
information beyond the conditioning bin itself,
\begin{equation}
  I(b_1:b_3\mid b_2,Y)=1\ \mathrm{bit}.
  \label{eq:app-Y-does-not-screen}
\end{equation}
The trivial one-dimensional auxiliary register also leaves the one-bit witness
unchanged.  Hence the minimal predictive memory has dimension two and
\begin{equation}
  m_\epsilon(L\mid R)=1\ \mathrm{bit}
  \quad(t<t_{\mathrm{Page}}),
  \qquad
  m_\epsilon(L\mid R)=0
  \quad(t>t_{\mathrm{Page}}),
  \label{eq:app-memory-cost-endpoints}
\end{equation}
for every tolerance below the pre-Page one-bit residual and above the
post-Page reconstruction floor.  This establishes the $k_\ast=3$ case of
Proposition~\ref{prop:memcost} and exhibits explicitly the hierarchy
$\log_2\dim M_{\min}=1<\log_2\dim M=2$.

\section{Operational and recovery proofs}
\label{app:proofs}

\subsection{Discrimination interpretation}
\label{app:dual}

For a fixed pair of combs, a tester $T$ maps each comb to a state or classical
outcome distribution.  With normalized trace distance as the base divergence,
\begin{equation}
  d_{\mathrm{op}}^{\mathrm{tr}}(\Upsilon,\Lambda)
  =\sup_{T\in\mathcal I}
    \frac12\big\|T[\Upsilon]-T[\Lambda]\big\|_1,
  \label{eq:app-trace-comb-distance}
\end{equation}
and equal prior probabilities give
\begin{equation}
  p_{\mathrm{succ}}^{\mathrm{opt}}(\Upsilon,\Lambda)
  =\frac12\left(1+d_{\mathrm{op}}^{\mathrm{tr}}(\Upsilon,\Lambda)\right).
  \label{eq:app-success-probability}
\end{equation}
For the relative-entropy base divergence,
\begin{equation}
  d_{\mathrm{op}}^{\mathrm{rel}}(\Upsilon,\Lambda)
  =\sup_{T\in\mathcal I}
    D\big(T[\Upsilon]\|T[\Lambda]\big)
  \label{eq:app-relative-comb-distance}
\end{equation}
is the optimal type-II exponent available to the tester family in the
corresponding asymptotic discrimination problem.

The memory is the infimum of either pairwise quantity over the factorized
Markov-comb class,
\begin{equation}
  \mathcal N[\Upsilon]
  =\inf_{\Lambda\in\mathrm{CPD}}
    d_{\mathrm{op}}(\Upsilon,\Lambda).
  \label{eq:app-memory-optimization}
\end{equation}
The reference class is nonconvex, so no minimax exchange is required or used.
In the finite-dimensional truncations used here, the parameter spaces of the
initial state and stepwise CPTP maps are compact and the objective is
continuous; the infimum is therefore attained. We nevertheless retain
$\inf$ in the general definitions. The optimizing Markov comb for
$\Upsilon_L$ and the optimizing assisted Markov comb for
$\Upsilon_{L|R}$ are independent.  Consequently
\begin{equation}
  \chi_{\mathrm{island}}
  =\inf_{\Lambda\in\mathrm{CPD}}
      d_{\mathrm{op}}(\Upsilon_L,\Lambda)
   -\inf_{\Lambda_R\in\mathrm{CPD}_R}
      d_{\mathrm{op}}(\Upsilon_{L|R},\Lambda_R)
  \label{eq:app-difference-of-optima}
\end{equation}
is precisely the reduction in optimized distinguishability from the
Markovian reference set produced by access to $R$.  Equation
\eqref{eq:app-difference-of-optima}, rather than a single common optimizing
tester, is the content of Proposition~\ref{prop:opmeaning}.

\subsection{Continuity bound for the forward bridge}
\label{app:constant}

Assumption A0 supplies a fixed side-leg channel
$\mathcal G_{M\to MR}$, independent of the late seeds, for which
\begin{equation}
  d_{\mathrm{op}}\!\left(
    \Upsilon_{L|MR},
    \mathcal G_{M\to MR}[\Upsilon_{L|M}]
  \right)
  \leq\delta_{\mathrm{med}}
  \label{eq:app-mediation}
\end{equation}
on the code subspace. This is the step that excludes correlations revealed by
$R$ but not mediated by $M$.

Let
\begin{equation}
  \widetilde\rho_{ML}
  =(\mathcal R_{R\to M}\otimes\mathrm{id}_L)(\rho_{RL})
  \label{eq:app-recovered-conditioning-state}
\end{equation}
be the recovered joint conditioning state. The convention in~\eqref{eq:A2}
is the root fidelity
$F(\rho,\sigma)=\|\sqrt\rho\sqrt\sigma\|_1$. If
$1-F(\rho_{ML},\widetilde\rho_{ML})\leq\epsilon(C)$, the Fuchs--van de
Graaf inequalities give
\begin{equation}
  \frac12\|\rho_{ML}-\widetilde\rho_{ML}\|_1
  \leq\sqrt{1-F(\rho_{ML},\widetilde\rho_{ML})^2}
  \leq\sqrt{2\epsilon(C)}.
  \label{eq:app-fvg}
\end{equation}

Let $\Gamma_{n:0}$ map the conditioning state to the $n$-slot late comb through
the fixed causal dilation. If the recovered register is inserted once and
then carried through the dilation, data processing gives
\begin{equation}
  d_{\mathrm{op}}\bigl(
    \Gamma_{n:0}[\rho_{ML}],
    \Gamma_{n:0}[\widetilde\rho_{ML}]
  \bigr)
  \leq\sqrt{2\epsilon(C)}
  \label{eq:app-global-continuity}
\end{equation}
for normalized trace-distance testers. A conservative local-use hybrid bound
replaces the true memory one link at a time and gives
\begin{equation}
  d_{\mathrm{op}}\bigl(
    \Upsilon_{L|M},
    \Upsilon_{L|R}^{\mathrm{rec}}
  \bigr)
  \leq n\sqrt{2\epsilon(C)}.
  \label{eq:app-telescoping-continuity}
\end{equation}
Thus $c_{\mathcal I,n}\leq\sqrt2 n$, with the sharper $\sqrt2$ constant when
A2 controls the complete joint conditioning state supplied once. Other base
divergences are represented by their continuity modulus
$\omega_{\mathcal I,n}$.

Choose a factorized Markov comb $\Lambda_M$ satisfying
\begin{equation}
  d_{\mathrm{op}}(\Upsilon_{L|M},\Lambda_M)
  \leq\delta_M+\eta,
  \qquad \eta>0.
  \label{eq:app-near-markov-M}
\end{equation}
Pre-compose its static side leg with $\mathcal R_{R\to M}$ and, where needed,
the fixed mediated extension $\mathcal G_{M\to MR}$. The resulting reference
$\Lambda_R^{\mathrm{rec}}$ lies in $\mathrm{CPD}_R$: side-leg pre-processing
cannot couple distinct late slots and the factorized late maps are unchanged.
Using~\eqref{eq:app-mediation}, the recovered-process continuity bound, and
\eqref{eq:app-near-markov-M},
\begin{align}
  \mathcal N_R(L)
  &\leq
  d_{\mathrm{op}}(\Upsilon_{L|R},\Lambda_R^{\mathrm{rec}})
  \\
  &\leq
  \delta_{\mathrm{med}}+\delta_M+\eta
  +\omega_{\mathcal I,n}(\epsilon(C)).
  \label{eq:app-forward-bridge-chain}
\end{align}
Taking $\eta\to0$ proves~\eqref{eq:T3bound}. Without A0, the first comparison
in~\eqref{eq:app-forward-bridge-chain} fails: a side register can expose
late-bin correlations even when $M$ is trivial and the late-only comb is
Markovian.

\subsection{Predictive randomization and the partial converse}
\label{app:simulator}

The minimal memory $M_{\min}$ is defined up to predictive equivalence: two
states are identified when they induce the same statistics for every future
tester in $\mathcal I$. Let
$\Phi_{\mathrm{fut}|M_{\min}}$ and $\Phi_{\mathrm{fut}|R}$ be the corresponding
predictive channels and define the randomization deficiency
\begin{equation}
  \Delta_{\mathcal I}(R\to M_{\min})
  =\inf_{\mathcal R_{R\to M_{\min}}}
   \big\|
      \Phi_{\mathrm{fut}|M_{\min}}\circ\mathcal R_{R\to M_{\min}}
      -\Phi_{\mathrm{fut}|R}
    \big\|_{\mathcal I_{\mathrm{fut}}}.
  \label{eq:app-predictive-deficiency}
\end{equation}
The channel set is compact in finite dimension, so the infimum is attained.

Let $\mathsf V_{\mathrm{pred}}$ be the finite-dimensional quotient space of
predictive channels modulo predictive equivalence. Its dimension is bounded
by a function of $d_{M_{\min}}$ and the finite future horizon
$n_{\mathrm{fut}}$. Informational completeness makes the future-tester
seminorm a norm on this quotient. Norm equivalence gives a constant
\begin{equation}
  \kappa_{\mathcal I}
  =\kappa_{\mathcal I}(d_{M_{\min}},n_{\mathrm{fut}})
  \label{eq:app-kappa-dependence}
\end{equation}
relating it to a fixed Euclidean norm; no dimension-independent value is
claimed.

The images
$\Phi_{\mathrm{fut}|M_{\min}}\circ\mathcal R_{R\to M_{\min}}$ form a compact
convex subset of $\mathsf V_{\mathrm{pred}}$. If every simulator failed by
more than $\eta$, Hahn--Banach separation would give a normalized future
decision functional whose advantage is at least $\eta/\kappa_{\mathcal I}$.
Assumptions B1 and B2 embed this functional into a centered past--future
memory witness $W$: all future-affecting information factors through
$M_{\min}$, and the centered functional obeys
\begin{equation}
  W[\Lambda_R]=0
  \qquad\text{for every }\Lambda_R\in\mathrm{CPD}_R.
  \label{eq:app-witness-annihilates-cpd}
\end{equation}
Thus the witness separates the assisted process from the entire factorized
Markov class, rather than from one selected reference.

Suppose the base operational divergence satisfies a Pinsker-type quadratic
lower bound
\begin{equation}
  d_{\mathrm{op}}(\Upsilon,\Lambda)
  \geq\alpha_{\mathcal I}
  \|\Upsilon-\Lambda\|_{\mathcal I}^2.
  \label{eq:app-quadratic-divergence}
\end{equation}
Relative entropy has this form after the tester optimization, and the same
square-root dependence follows from approximate-sufficiency recovery
estimates~\cite{FawziRenner,JungeUniversal}. Equations
\eqref{eq:app-witness-annihilates-cpd} and
\eqref{eq:app-quadratic-divergence} imply
\begin{equation}
  \mathcal N_R(L)
  \geq
  \alpha_{\mathcal I}
  \frac{\Delta_{\mathcal I}(R\to M_{\min})^2}
       {\kappa_{\mathcal I}(d_{M_{\min}},n_{\mathrm{fut}})^2}.
  \label{eq:app-deficiency-lower-bound}
\end{equation}
Consequently,
\begin{equation}
  \mathcal N_R(L)\leq\epsilon
  \quad\Longrightarrow\quad
  \Delta_{\mathcal I}(R\to M_{\min})
  \leq
  \kappa_{\mathcal I}(d_{M_{\min}},n_{\mathrm{fut}})
  \sqrt{\epsilon/\alpha_{\mathcal I}}.
  \label{eq:app-randomization-bound}
\end{equation}
For normalized trace-distance memory, the corresponding separation gives a
linear modulus instead of the square-root modulus. Selecting the minimizing
channel in~\eqref{eq:app-predictive-deficiency} gives the simulator in
Theorem~\ref{thm:converse}.

The statement concerns influence rather than microscopic state recovery. If
the map from memory states to future statistics is injective on the code
subspace, predictive equivalence is trivial and the simulator reconstructs
the state of $M_{\min}$. In general it reconstructs the unique predictive
class with the required future action.

\subsection{Predictive compression versus microscopic recovery}
\label{app:predictivecompression}

A simple extension of the screening core illustrates the distinction.  Let
$M=Z\otimes Q$, where the pointer $Z$ controls all late emissions and the
spectator $Q$ never enters the predictive channel.  The early register contains
an exact copy of $Z$, while $Q$ remains maximally entangled with an inaccessible
reference $A_Q$:
\begin{equation}
  |\Psi\rangle
  =\frac1{\sqrt2}\sum_{z=0}^{1}
    |z\rangle_{A_Z}|z\rangle_Z|z\rangle_R
    |z\rangle_{b_1}|z\rangle_{b_2}
    \otimes|\Phi_{d_Q}\rangle_{QA_Q}.
  \label{eq:app-pointer-spectator-state}
\end{equation}
The minimal predictive memory is $M_{\min}=Z$.  Reading $Z$ from $R$ reproduces
all future statistics exactly, while no channel on $R$ can recover the
entanglement between $Q$ and $A_Q$.  For every product state
$\rho_{\widehat Q}\otimes\sigma_{A_Q}$, the root fidelity obeys
\begin{equation}
  F\bigl(
    |\Phi_{d_Q}\rangle,
    \rho_{\widehat Q}\otimes\sigma_{A_Q}
  \bigr)
  \leq\frac1{\sqrt{d_Q}}.
  \label{eq:app-spectator-fidelity}
\end{equation}
In the actual channel construction the inaccessible marginal remains
$\tau_{A_Q}=\mathbf 1/d_Q$, for which the sharper value is $1/d_Q$.
This example is not a strict process-wedge inclusion theorem, because the
operational object is the predictive class $[M]=[Z]$.  It instead explains why
the converse and the memory cost are formulated using $M_{\min}$: future
simulation can be exact even when an enlarged microscopic register contains
spectator degrees of freedom that remain unrecoverable.

\section{Numerical methods and certification}
\label{app:numerics}

\subsection{Tester families and finite-dimensional parameterization}
\label{app:testerfamilies}

The numerical model represents each late bin by a qubit and each instrument by
Choi matrices satisfying positivity and trace-preservation constraints.  The
three families used in Figure~\ref{fig:robustness}(a) are:
\begin{enumerate}
  \item the computational detector, whose projective effects are aligned with
        the controlled-emission basis;
  \item the tilted detector, obtained by a fixed single-qubit rotation before
        the same readout;
  \item the unrestricted finite-dimensional family, optimized over the
        admitted instrument parameterization and retained ancillas.
\end{enumerate}
For a fixed tester, linking its Choi operator with the comb produces a finite
outcome distribution.  The passive trace-distance objective is therefore a
total-variation distance.  Coherent instruments retain the same linear link
product, followed by the trace norm of the final tester state.

\subsection{Distances, lower bounds, and numerical brackets}
\label{app:numericalbrackets}

The optimization contains an infimum over factorized Markov references and a
supremum over testers, so pairwise and memory bounds must be distinguished. A
feasible Markov comb $\Lambda$ gives
\begin{equation}
  \mathcal N[\Upsilon]
  \leq d_{\mathrm{op}}(\Upsilon,\Lambda),
\end{equation}
and a certified optimization over testers for that fixed reference sharpens
this upper bound. By contrast, a lower bound on $\mathcal N$ requires either a
witness functional that vanishes on every element of $\mathrm{CPD}$, such as
the parity statistic, or a certified minimization over the complete reference
parameterization. A multistart tester search improves a lower bound on a
pairwise distance; it becomes a lower bound on the memory only when accompanied
by such a universal Markov witness or a certified reference minimization.

We report a bracket $[L,U]$ only when
\begin{equation}
  L\leq\mathcal N\leq U,
  \qquad
  U-L\leq\epsilon_{\mathrm{num}},
  \label{eq:app-bracket-condition}
\end{equation}
and the witness, relaxation, covering, or reference-minimization certificate is
stored with the data. For grid-certified curves, let $Q_h$ be the value on
mesh spacing $h$. The reported discretization error is bounded by
\begin{equation}
  \epsilon_{\mathrm{grid}}
  =|Q_h-Q_{h/2}|+L_Qh/2,
  \label{eq:app-grid-error}
\end{equation}
where $L_Q$ is a Lipschitz bound for the objective on the final cell. Closed
form curves, including the erasure path and parity-order step, have zero
discretization error.

\subsection{Code-subspace and temporal-order scans}
\label{app:scans}

The recovery scan uses
$d_{R,\mathrm{eff}}(t)=2^{\rho t}$ and the overload
\begin{equation}
  \mu(t,C)=\rho t-\log_2d_{M_{\min}}-\log_2\dim C.
  \label{eq:app-overload-scan}
\end{equation}
The reconstruction error and decoder coordinate are then evaluated from
\eqref{eq:recovery-calibration}.  The threshold
\begin{equation}
  t_{\mathrm{dec}}(C)
  =\frac{\log_2d_{M_{\min}}+\log_2\dim C}{\rho}
  \label{eq:app-decoding-threshold}
\end{equation}
is not fitted separately for each curve.  Increasing the code subspace changes
only the overload and the corresponding recovery error.

For order scans, the reduced combs are constructed by tracing every subset of
slots outside the selected span.  In the order-$K_0$ parity family,
\eqref{eq:app-higher-order-separation} provides an analytic certificate that
all $K<K_0$ values vanish.  At $K_0$, evaluating the parity statistic supplies
a nonzero lower bound, and the known running-parity dilation fixes the
saturation order.  For a general numerical comb the criterion is instead the
windowed residual~\eqref{eq:order-residual}; both constituent memories are
checked separately before their difference is declared saturated.

\subsection{Conditional memory-cost certification}
\label{app:memorycostnumerics}

A candidate auxiliary register $M'$ and a channel inserting it into the late
dilation provide an upper bound on the infimum in~\eqref{eq:mcost}.  For each
candidate dimension, the residual assisted memory is computed at the
saturated tester order.  The smallest dimension satisfying
\begin{equation}
  \mathcal N_{RM'}^{(k_\ast)}(L)
  \leq\epsilon_\ast
  \label{eq:app-memory-cost-test}
\end{equation}
is reported.  The tolerance obeys
\eqref{eq:memory-cost-tolerance}, so the decision is separated from solver,
order, and reconstruction floors.

For deterministic parity processes, the one-bit upper bound is attained.  The trivial one-dimensional extension leaves the witness in
\eqref{eq:app-one-bit-witness} above every selected
$\epsilon_\ast<1$, while the one-bit running parity makes the conditional
witness vanish.  For the common-cause family, supplying the known minimal
dilation gives the upper bound $\log_2d_{M_{\min}}$ before recovery and zero
after the decoded $R$ crosses the tolerance.  These analytic checks accompany
the numerical scans in Figure~\ref{fig:robustness}(d).

\subsection{Reproducibility protocol}
\label{app:reproducibility}

Each plotted data set records the model dimension, the instrument family, the
base divergence, the tester span, the code-subspace dimension, the recovery
calibration, and the numerical provenance.  Reproducing a curve consists of:
\begin{enumerate}
  \item constructing $\mathsf E_{\lambda,C}$ at each recovery coordinate;
  \item evaluating the late-only and assisted combs under the same tester
        family;
  \item minimizing over the factorized Markov reference class;
  \item repeating the calculation at increasing tester order until
        \eqref{eq:order-residual} is satisfied;
  \item storing the lower and upper certificates together with the raw
        optimizer output.
\end{enumerate}
The model code and figure data are released with these metadata so that
every closed-form, grid-certified, and bracketed curve can be regenerated
without reconstructing choices from the manuscript text.

\section{Semiclassical bookkeeping}
\label{app:semibook}

\subsection{Exterior and island branches}
\label{app:semibranches}

For a radiation region $X$, the adiabatic bookkeeping used in
Section~\ref{sec:semiclassical} assigns the two generalized-entropy branches
\begin{align}
  S_{\mathrm{gen}}^{\mathrm{ext}}(X;u)
  &=\Sigma_{\mathrm{rad}}(X;u)+H_{\mathrm{rec}}(X),
  \label{eq:app-ext-branch}\\
  S_{\mathrm{gen}}^{\mathrm{isl}}(X;u)
  &=S_{\mathrm{BH}}(u_{\mathrm{last}})
    +\Sigma_{\mathrm{rad}}(\overline X;u)+H_M.
  \label{eq:app-isl-branch}
\end{align}
Here $\Sigma_{\mathrm{rad}}(X;u)$ is the coarse radiation entropy in the
selected intervals, $\Sigma_{\mathrm{rad}}(\overline X;u)$ is the partner
entropy assigned to the complementary emitted bins, $H_{\mathrm{rec}}(X)$
counts diary imprints retained in the exterior branch, and $H_M$ is the
diary-reference entropy.  The precise interval assignments follow the
arbitrary-subset prescription of
Refs.~\cite{Hollowood:2020cou,Hollowood:2020couPage}.

For a bin
$b_i=[u+(i-1)\Delta u,u+i\Delta u]$, the entropy flux in the exponential
background is
\begin{equation}
  F_i(u)
  =S_{\mathrm{BH}}\bigl(u+(i-1)\Delta u\bigr)
   -S_{\mathrm{BH}}\bigl(u+i\Delta u\bigr).
  \label{eq:app-bin-flux}
\end{equation}
Substituting the appropriate sums of $F_i$ into
\eqref{eq:app-ext-branch} and \eqref{eq:app-isl-branch} gives the four gaps
$\Delta_X$ in~\eqref{eq:gapdef}.  Their zeros determine the exchange times
used in Figure~\ref{fig:jtwindow}.

\subsection{Compatible islands and geometric cancellation}
\label{app:compatibleislands}

Write the witness as the signed sum
\begin{equation}
  W_R^{(3)}=\sum_{X\in\mathcal X_3}\sigma_X S(X),
  \qquad
  \sigma_X=
  \begin{cases}
    +1,&X=Rb_1b_2,\ Rb_2b_3,\\
    -1,&X=Rb_2,\ Rb_1b_2b_3.
  \end{cases}
  \label{eq:app-witness-signs}
\end{equation}
When the four selected saddles contain a common island $I_\star$, the area and
endpoint terms form the same inclusion-exclusion combination as the radiation
regions.  Their exact cancellation gives
\begin{equation}
  W_R^{(3)}
  =I_{\mathrm{bulk}}(b_1:b_3\mid b_2RI_\star).
  \label{eq:app-compatible-reduction}
\end{equation}
A mismatch of island endpoints or of the selected replica state contributes a
signed residual.  Its magnitude is the compatibility defect
$\eta_{\mathrm{QES}}^{(3)}$ appearing in
\eqref{eq:semired}.  Applying the computable recovery bound to the bulk state
then gives~\eqref{eq:semibridge}.

\subsection{Sharp saddles and the gap identity}
\label{app:sharpsaddles}

Let $[x]_+=\max\{x,0\}$.  Since
$S_{\mathrm{gen}}^{\mathrm{dom}}=S_{\mathrm{gen}}^{\mathrm{ext}}-[\Delta_X]_+$,
substitution into~\eqref{eq:app-witness-signs} gives
\begin{equation}
  W_{R,\mathrm{sharp}}^{(3)}(u)
  =W_R^{(3),\mathrm{ext}}(u)
   -\sum_{X\in\mathcal X_3}\sigma_X[\Delta_X(u)]_+.
  \label{eq:app-sharp-witness}
\end{equation}
The gaps obey the inclusion-exclusion identity
\begin{equation}
  \sum_{X\in\mathcal X_3}\sigma_X\Delta_X(u)
  =W_R^{(3),\mathrm{ext}}(u)-W_R^{(3),\mathrm{isl}}(u).
  \label{eq:app-gap-identity}
\end{equation}
For the parity diary in the exterior regime,
$W_R^{(3),\mathrm{ext}}=1$ bit.  A negatively signed region can exchange first
in a decreasing-flux background; the corresponding term in
\eqref{eq:app-sharp-witness} raises the separately minimized expression.  This
is the origin of the sharp-saddle overshoot.  The identity also makes clear
why separate minimization is sensitive to saddle compatibility: the four
positive-part operations need not correspond to one joint replicated
geometry.

\subsection{Soft-min interpolation and the defect bound}
\label{app:softmin}

A two-saddle interpolation assigns
\begin{equation}
  S_X^{\mathrm{soft}}
  =-\log_2\!
   \left(2^{-S_X^{\mathrm{ext}}}+2^{-S_X^{\mathrm{isl}}}\right)
  =\min\{S_X^{\mathrm{ext}},S_X^{\mathrm{isl}}\}
   -\log_2\!\left(1+2^{-|\Delta_X|}\right).
  \label{eq:app-soft-min}
\end{equation}
The absolute change of the four-entropy combination relative to the sharp
choice is therefore bounded by
\begin{equation}
  \eta_{\mathrm{QES}}^{(3)}(u)
  \leq\sum_{X\in\mathcal X_3}
    \log_2\!\left(1+2^{-|\Delta_X(u)|}\right)
  \leq4\log_2\!\left(1+2^{-\Delta_{\min}(u)}\right),
  \label{eq:app-soft-defect}
\end{equation}
where $\Delta_{\min}=\min_X|\Delta_X|$.  This interpolation quantifies finite
gap sensitivity.  It does not impose compatibility among the four replica
geometries, so the joint replica calculation remains the physical evaluation
of the crossover.

\subsection{Background and code-subspace checks}
\label{app:semiclassicalchecks}

The entropy-conserving exponential model uses
\begin{equation}
  S_{\mathrm{BH}}(u)=S_0e^{-u/\tau},
  \qquad
  S_{\mathrm{rad}}(u)=S_0-S_{\mathrm{BH}}(u).
  \label{eq:app-exponential-background}
\end{equation}
The numerical check against the exact AEMM boundary reparametrization compares
three quantities over the Page-regime window: the boundary trajectory, the
coarse entropy flux, and the QES scrambling lag.  The exponential model is
used only where these quantities remain within the stated numerical tolerance;
the exact solution supplies the systematic background check rather than an
additional fit parameter.

If a code-subspace contribution $\log_2\dim C$ enters a gap additively, the
implicit equation $\Delta_X(u_X^\ast,C)=0$ gives
\begin{equation}
  \frac{\partial u_X^\ast}{\partial\log_2\dim C}
  =-\frac{\partial_{\log_2\dim C}\Delta_X}
          {\partial_u\Delta_X(u_X^\ast,C)}.
  \label{eq:app-general-code-shift}
\end{equation}
With the sign convention used in Section~\ref{sec:semiclassical},
$\partial_{\log_2\dim C}\Delta_X=-1$, which reduces
\eqref{eq:app-general-code-shift} to~\eqref{eq:codeshift}.  This is the
semiclassical counterpart of the overload shift in
\eqref{eq:app-decoding-threshold}.

\section{Predictive influence and the process wedge}
\label{app:influence}

\subsection{Predictive equivalence and order localization}
\label{app:orderlocalization}

Definition~\ref{def:pw} uses substitution-level predictive equivalence: two
registers are equivalent when their slot-to-future maps agree for every
admissible preparation and every $\rho\in C$. A stronger
intervention-resolved relation declares $X\sim_K^{\mathrm{int}}X'$ when
\begin{equation}
  d_{\mathrm{op}}^{(K)}\!\left(
    \Upsilon_L[\mathcal J_X;\rho],
    \Upsilon_L[\mathcal J_{X'};\rho]
  \right)=0
  \label{eq:app-influence-equivalence}
\end{equation}
for every admissible intervention $\mathcal J$ and every $\rho\in C$. Because
the intervention family contains the trivial operation and preparation maps,
intervention-resolved equivalence refines the substitution relation and every
intervention-resolved simulation bound implies membership in the wedge of
Definition~\ref{def:pw}.

\begin{lemma}[Order localization]
\label{lem:orderloc}
If
$d_{\mathrm{op}}^{(K)}(\Upsilon_1,\Upsilon_2)\leq\varepsilon$, then every
tester that distinguishes the two processes by more than $\varepsilon$ spans
more than $K$ late slots.
\end{lemma}

\emph{Proof.}
The definition of $d_{\mathrm{op}}^{(K)}$ takes the supremum over every tester
with span at most $K$. A tester in that class with divergence exceeding
$\varepsilon$ would contradict the assumed bound. \hfill$\square$

The lemma gives the filtration its direct operational meaning. Simulation at
resolution $K$ leaves any residual difference above the selected tolerance
accessible only to longer experiments.

\subsection{Composition of reconstruction with influence}
\label{app:ewinclusion}

Let a register $X$ be reconstructible from $R$ on $C$ through
$\mathcal R_{R\to X}$ with infidelity $\epsilon_X(C)$.  For each intervention
$\mathcal J_X$, apply the same operation to the reconstructed register
$\widehat X$.  The Fuchs and van de Graaf bound followed by data processing
through the intervention and the $K$-slot late process gives
\begin{equation}
  \sup_{\mathcal J}\sup_{\rho\in C}
  d_{\mathrm{op}}^{(K)}\!\left(
    \Upsilon_L[\mathcal J_X;\rho],
    \Upsilon_L[\mathcal J_{\widehat X};\rho]
  \right)
  \leq c_{\mathcal I,K}\sqrt{\epsilon_X(C)}.
  \label{eq:app-ew-influence-bound}
\end{equation}
Thus the quotient image $q_K([X])$ belongs to the process wedge at the
corresponding tolerance, proving the inclusion used in Definition~\ref{def:pw}. The
admissible-intervention family contains the trivial and preparation cases, so
\eqref{eq:app-ew-influence-bound} implies the substitution criterion in
Definition~\ref{def:pw}. The inclusion is uniform on the code subspace because
one recovery channel is used for every $\rho\in C$.

A register whose influence is invisible below order $K_0$ belongs to the
trivial class at every $K<K_0$.  The parity family supplies the canonical
example: all proper marginals coincide with the Markov reference, while the
class separates at $K_0$.  This finite-resolution coarsening is an order
statement, not a claim that the microscopic state is reconstructible.

\subsection{Full resolution and predictive faithfulness}
\label{app:fullresolution}

At full resolution, an informationally complete instrument family determines
the entire influence comb of a register. The inclusion from left to right is
the zero-reconstruction-error case of Appendix~\ref{app:ewinclusion}. For the converse, predictive faithfulness
means that distinct states of $M_{\min}$ induce distinct full-resolution
future functionals on $C$; the injectivity argument closing
Appendix~\ref{app:simulator} then turns exact influence simulation into state
reconstruction up to predictive equivalence. Consequently, in the exact
faithful core,
\begin{equation}
  q_\infty\!\left(EW(R;C)\right)
  =\mathrm{PW}^{(\infty),0}(R;C,\mathcal I)
  \label{eq:app-faithful-core}
\end{equation}
when both sides are read as predictive classes and assumptions B1 and B2 hold.
This equality is a conditional statement about the faithful core.  Restricted
instruments, finite tolerance, spectator sectors, and independent side
channels replace it by the established inclusion
\eqref{eq:app-ew-influence-bound}.

\section{Algebraic bridge and notation}
\label{app:algebraic}

\subsection{Finite-dimensional algebraic bridge}
\label{app:algebraicbridge}

Let $\mathcal A$ be a finite-dimensional von Neumann algebra, with predictive
island subalgebra $\mathcal A_I$ and early-radiation subalgebra
$\mathcal A_R$.  A state-preserving conditional expectation
$E_R:\mathcal A\to\mathcal A_R$ specifies the accessible coarse-graining.  A
recovery morphism
$\beta^\dagger:\mathcal A_I\to\mathcal A_R$ and its Schr\"odinger-picture
predual $\beta_*$ specify reconstruction.  These are distinct maps.

For a subalgebra $\mathcal B$, let $\mathcal N_{\mathcal B}(L)$ denote the
operational memory when the tester may process $\mathcal B$ on the initial
side leg. Assume algebraic mediation: conditioned on $\mathcal A_I$, the
early algebra carries no additional late-process correlation, up to
$\delta_{\mathrm{med}}$. Assume also
\begin{equation}
  \mathcal N_{\mathcal A_I}(L)\leq\delta_M
  \label{eq:app-algebraic-sufficiency}
\end{equation}
and that one reconstruction map works uniformly on $C$ in the operational
norm relevant to the late process,
\begin{equation}
  \sup_{\rho\in C}\sup_{\mathcal J,T}
  D\!\left(
    T[\Upsilon_L(\mathcal J_{\mathcal A_I};\rho)],
    T[\Upsilon_L(\mathcal J_{\beta^\dagger(\mathcal A_I)};\rho)]
  \right)
  \leq\epsilon(C).
  \label{eq:app-algebraic-recovery}
\end{equation}
Here the suprema run over the admitted island interventions and late testers.
Replacing the true island intervention by its reconstructed representative and
using data processing gives
\begin{equation}
  \mathcal N_{\mathcal A_R}(L)\big|_C
  \leq\delta_{\mathrm{med}}+\delta_M
  +\omega_{\mathcal I,n}(\epsilon(C)).
  \label{eq:app-algebraic-bridge}
\end{equation}
This is the algebraic form of the forward bridge.  In a tensor product
$\mathcal A=\mathcal B(\mathcal H_R)\overline\otimes
\mathcal B(\mathcal H_{\bar R})$, a state-preserving expectation acts
schematically as
\begin{equation}
  E_R(a_R\otimes a_{\bar R})
  =a_R\,\omega(a_{\bar R})\otimes\mathbf 1_{\bar R},
  \label{eq:app-tensor-expectation}
\end{equation}
while $\beta_*$ is the separate Petz-type recovery channel.  Equation
\eqref{eq:app-algebraic-bridge} therefore reduces to Theorem~\ref{thm:bridge}
without identifying coarse-graining with recovery.

Petz sufficiency is naturally formulated for subalgebras~\cite{Petz}.
Crossed-product constructions provide type-$\mathrm{II}$ algebras with traces
and generalized entropies in gravitational subregions
\cite{LeutheusserLiu,WittenCrossed,CPW}.  Extending
\eqref{eq:app-algebraic-recovery} to type-$\mathrm{II}$ and type-$\mathrm{III}$
settings requires an energy-bounded class of local testers and a continuity
estimate in a completely bounded or operational norm.  The finite-dimensional
statement above fixes the object that such an extension must control.

\subsection{Notation dictionary}
\label{app:notation}

\begin{table}[t]
  \centering
  \footnotesize
  \renewcommand{\arraystretch}{1.08}
  \begin{tabular}{p{0.18\textwidth}p{0.70\textwidth}}
    \hline
    Symbol & Meaning \\
    \hline
    $b_i$, $L$ & Individual late radiation bins and their ordered collection. \\
    $R$ & Early radiation retained as initial side information. \\
    $M$ & A sufficient interior memory in a chosen dilation. \\
    $M_{\min}$ & Minimal predictive memory, unique up to predictive equivalence. \\
    $A$ & External reference that keeps the interior record mixed. \\
    $C$, $\epsilon(C)$ & Code subspace and uniform reconstruction error on it. \\
    $F(\rho,\sigma)$ & Root fidelity, $\|\sqrt\rho\sqrt\sigma\|_1$; infidelity means $1-F$. \\
    $\lambda(t,C)$ & Recovery coordinate calibrated by retarded time and code-subspace size. \\
    $\delta_{\mathrm{med}}$, $\delta_M$ & Mediation residual and memory remaining after conditioning on the interior; $\delta=\delta_{\mathrm{med}}+\delta_M$. \\
    $\tau_\beta$ & Thermal single-bin reference state. \\
    $\mathcal I$ & Admissible instrument and tester family. \\
    $\Upsilon_L$, $\Upsilon_{L|R}$ & Late-only and early-radiation-assisted process tensors. \\
    $\mathrm{CPD}$, $\mathrm{CPD}_R$ & Factorized Markov-comb reference classes without and with the static side register. \\
    $d_{\mathrm{op}}^{(K)}$ & Tester-based process distance restricted to span at most $K$. \\
    $\mathcal N_\varnothing$, $\mathcal N_R$ & Late-only and assisted temporal memories. \\
    $\chi_{\mathrm{island}}$ & Memory removed by access to $R$, $\mathcal N_\varnothing-\mathcal N_R$. \\
    $m_\epsilon(L|R)$ & Smallest additional predictive memory that Markovianizes $L$ to tolerance $\epsilon$. \\
    $k_\ast$ & Saturation order; in the single-order family it is also the first visible order. \\
    $\epsilon_{\mathrm{num}}$, $\epsilon_{\mathrm{trunc}}(K)$, $\epsilon_\ast$ & Numerical floor, windowed order residual, and memory-cost tolerance. \\
    $W_\varnothing^{(3)}$, $W_R^{(3)}$ & Passive order-three conditional-mutual-information witnesses. \\
    $\Delta_X$, $\eta_{\mathrm{QES}}^{(3)}$ & Generalized-entropy saddle gap for region $X$ and the four-entropy compatibility defect. \\
    $H_M$ & Entropy of the interior diary or record relative to its reference. \\
    $\mathcal A_I$, $\mathcal A_R$ & Predictive island algebra and early-radiation algebra. \\
    $\mathrm{PW}^{(K),\epsilon}$ & Process wedge of predictive influence classes at fixed resolution and tolerance. \\
    $EW(R)$ & Entanglement wedge, understood relative to the same code subspace and recovery tolerance. \\
    \hline
  \end{tabular}
  \caption{Notation used throughout the manuscript.  Every process quantity is
  relative to the stated instrument family, code subspace, temporal resolution,
  and tolerance when those labels are suppressed.}
  \label{tab:notation}
\end{table}

\end{document}